\documentclass[a4paper]{article}
\usepackage[a4paper, total={6in, 8.5in}]{geometry}

\newif\ifarXiv
\arXivtrue

\usepackage{amsthm}
\usepackage[table]{xcolor}
\usepackage{auto-pst-pdf}
\usepackage{amssymb}
\usepackage{amsmath}
\usepackage{mathrsfs}
\usepackage{stmaryrd}
\usepackage{times}
\usepackage{latexsym}
\usepackage{hhline}
\usepackage{enumerate} %per avere numeri romani ecc negli enumerate
\usepackage{calc} % From LaTeX distribution
\usepackage{graphicx} % From LaTeX distribution
\usepackage{ifthen} % From LaTeX distribution
\usepackage{pst-poly} % From pstricks/contrib/pst-poly
\usepackage{multido} % From PSTricks
\usepackage[normalem]{ulem} %Per \sout
\usepackage[textwidth=4cm,textsize=footnotesize]{todonotes}

\usepackage{enumitem}
\usepackage{eucal} %per fare un mathcal diverso
\usepackage{arydshln}
\usepackage{multirow}
\usepackage{relsize}%per testo piu' grande in mathmode: \mathlarger{ }
\usepackage[table]{xcolor}% per colorare le celle

\newcommand{\codeX}[1]{ \llbracket #1 \rrbracket_X}
\newcommand{\code}[2]{ \llbracket #2 \rrbracket_{#1}}

\newcommand{\ktile}[2]{B_{#2,#2}\left(#1\right)}
\newcommand{\tile}[1]{\ktile{#1}{2}}
\newcommand{\tset}{{\mathsf T}}
\newcommand{\tsetM}{{\mathsf M}}

\newcommand{\kblock}[2]{P_{#2,#2}\left(#1\right)}
\newcommand{\mytile}[4]{{ \scalebox{0.8}{
\begin{tabular}{|p{0.25cm}p{0.25cm}|}\cline{1-2}
$#1$ & $#2$ \\ $#3$ & $#4$\\ \cline{1-2}
\end{tabular}
}}}

\newcommand{\mytileDynamic}[4]{{ \scalebox{0.8}{
\begin{tabular}{|ll|}\cline{1-2}
$#1$ & $#2$ \\ $#3$ & $#4$\\ \cline{1-2}
\end{tabular}
}}}
\newcommand{\myLargeTile}[9]{{ \scalebox{0.8}{
\begin{tabular}{|lll|}\cline{1-3}
$#1$ & $#2$ & $#3$ \\ $#4$ & $#5$ & $#6$ \\ $#7$ & $#8$ & $#9$ \\ \cline{1-3}
\end{tabular}
}}}
\newcommand{\Mod}[1]{\ (\mathrm{mod}\ #1)} %operatore mod senza spazio aggiunto
\theoremstyle{plain}
\newtheorem{definition}{Definition}

\newtheorem{proposition}{Proposition}
\newtheorem{lemma}{Lemma}
\newtheorem{theorem}{Theorem}
\newtheorem{corollary}{Corollary}

\theoremstyle{remark}
\newtheorem{example}{Example}
\newtheorem{remark}{Remark}

\makeatletter
\DeclareFontFamily{OMX}{MnSymbolE}{}
\DeclareSymbolFont{MnLargeSymbols}{OMX}{MnSymbolE}{m}{n}
\SetSymbolFont{MnLargeSymbols}{bold}{OMX}{MnSymbolE}{b}{n}
\DeclareFontShape{OMX}{MnSymbolE}{m}{n}{
	<-6> MnSymbolE5
	<6-7> MnSymbolE6
	<7-8> MnSymbolE7
	<8-9> MnSymbolE8
	<9-10> MnSymbolE9
	<10-12> MnSymbolE10
	<12-> MnSymbolE12
}{}
\DeclareFontShape{OMX}{MnSymbolE}{b}{n}{
	<-6> MnSymbolE-Bold5
	<6-7> MnSymbolE-Bold6
	<7-8> MnSymbolE-Bold7
	<8-9> MnSymbolE-Bold8
	<9-10> MnSymbolE-Bold9
	<10-12> MnSymbolE-Bold10
	<12-> MnSymbolE-Bold12
}{}

\let\llangle\@undefined
\let\rrangle\@undefined
\DeclareMathDelimiter{\llangle}{\mathopen}%
{MnLargeSymbols}{'164}{MnLargeSymbols}{'164}
\DeclareMathDelimiter{\rrangle}{\mathclose}%
{MnLargeSymbols}{'171}{MnLargeSymbols}{'171}
\makeatother

\newcommand{\subd}[4]{(#1,#2\ ;\ #3,#4)}

\def\rec{\text{\rm{REC}}} % rec

 {\setlength\celldim{#1}%
 \setlength\extraheight{\celldim - \fontheight}%
 \setcounter{sqcolumns}{#2 - 1}%
 \begin{tabular}
 {|S|*{ \value{sqcolumns} }{Z|}}
 \hline}
 {\end{tabular}}

\begin{document}
\ifarXiv \title{
Reducing the local alphabet size in tiling systems by means of 2D comma-free codes}
\author{Stefano {Crespi Reghizzi}$^1$ \and Antonio Restivo$^2$ \and Pierluigi {San Pietro}$^1$}
\date{
	$^1$Dipartimento di Elettronica Informazione e Bioingegneria, \\Politecnico di Milano,\\Piazza Leonardo da Vinci 32, Milano, Italy 
	\\ \texttt{\{stefano.crespireghizzi, pierluigi.sanpietro\}@polimi.it}\\%
	$^2$Dipartimento di Matematica e Informatica, \\Universit\`{a} di Palermo  \\ \texttt{antonio.restivo@unipa.it}\\[2ex]%
}

\maketitle
\else
\title{
Reducing the local alphabet size in tiling systems by means of 2D comma-free codes
\tnoteref{t1}\tnoteref{t2}}
 %\tnotetext[t1]{This research was partially supported by *****.}
\tnotetext[t2]{Part of this work was published preliminarily in~\cite{Crespi-ReghizziRestivoSanPietro21} and presented at the conference \emph{DLT 2021}.}

\author[pdm,cnr]{Stefano {Crespi Reghizzi}}
\ead{stefano.crespireghizzi@polimi.it}
\author[unipa]{Antonio Restivo}
\ead{antonio.restivo@unipa.it}
\author[pdm,cnr]{Pierluigi {San Pietro}}
\ead{pierluigi.sanpietro@polimi.it}

\address[pdm]{Dipartimento di Elettronica Informazione e Bioingegneria, Politecnico di Milano,\\Piazza Leonardo da Vinci 32, Milano, Italy}
\address[cnr]{CNR IEIIT-MI, Milano, Italy}
\address[unipa]{Dipartimento di Matematica e Informatica, Universit\`{a} di Palermo }

\fi

\begin{abstract}
A recognizable picture language is defined as the projection of a local picture language defined by a set of two-by-two tiles, i.e. by a strictly-locally-testable (SLT) language of order 2. 
The family of recognizable picture languages is also defined, using larger $k$ by $k$ tiles, $k>2$, by the projection of the corresponding SLT language. 
A basic measure of the descriptive complexity of a picture language is given by the size of the SLT alphabet using two-by-two tiles, more precisely by the so-called alphabetic ratio of sizes: 
SLT-alphabet / picture-alphabet. We study how the alphabetic ratio changes moving from tiles of size two to tiles of larger size, and we obtain the following result: 
any recognizable picture language over an alphabet of size $n$ is the projection of an SLT language over an alphabet of size $2n$. 
 Moreover, two is the minimal alphabetic ratio possible in general. 
{ The proof  relies on a new family of comma-free picture codes, for which a lower bound on  numerosity is established; and on the relation of languages of encoded pictures with SLT languages.} 
Our result reproduces in two dimensions a similar property (known as Extended Medvedev's theorem) of the regular word languages, 
concerning the minimal alphabetic ratio needed to define a language by means of a projection of an SLT word language.
 
\end{abstract}
\ifarXiv 
\else \maketitle %\date{\today}
\fi
\pagenumbering{arabic}

\section{Introduction}\label{sect:introd}
 The \emph{tiling recognizable}, for short \rec, \ picture languages~\cite{GiammRestivo1997} are one of the best known extensions of the formal languages from one to two dimensions, i.e., from words to digital pictures (see~\cite{DBLP:books/ems/21/Crespi-ReghizziGL21} for a recent survey). An element of a picture language is a rectangular array of cells each one containing a symbol, called pixel, from a finite alphabet.
\par
To see how the formal definition of \rec~extends the definition of regular word languages (for short REG) we have to consider the classic characterization by the projection (in the sense of a letter-to-letter morphism) of local word languages. 
This definition is also known as Medvedev's theorem~\cite{Medvedev1964,Eilenberg74}, for short {MT}.
A local language is simply defined by the set words of length two that may occur as subtrings. Then, Medvedev's theorem says that, for each regular language $R$ over an alphabet $\Sigma$, there exists a local language $L$ over an alphabet $\Lambda$ and a letter-to-letter projection $h:\,\Lambda^* \to \Sigma^*$ such that $R=h(L)$. The two alphabets $\Sigma$ and $\Lambda$ are respectively called terminal and local. 
A word $z$ of the local language is called a pre-image of the terminal word $h(z)$. 
\par
An analogous statement is the primary definition of the \rec~languages. It differs with respect to the local languages used in the definition that now are the local picture languages. 
These are defined by the set of two-by-two pictures they may contain; each such small square picture contains therefore four pixels and is called a \emph{tile}. 
As a matter of fact, the formal definition of \rec~is based on a so-called \emph{tiling system} (TS) consisting of a set of tiles over the local alphabet and of a projection from the local alphabet to the terminal one. 
\par
Like the REG family, the \rec~family too can be defined by different formal approaches in particular by a sort of 2D cellular automaton. 
But such alternative definitions do not have the same fundamental role for pictures as, say, the finite-state model for words.
\par
Quite a few formal properties of REG languages continue to hold for \rec,~in particular the closure properties under projection, concatenation, union and intersection. 
The present research adds to the list of such properties a new one that pertains to the MT definition of both REG and \rec~languages, and we may call the ``extended Medvedev's theorem with alphabetic ratio two''. It is easier to explain what this property is starting from its formulation for REG and then lifting it to \rec.
\par
Local word languages, which as we mentioned are the pre-image $L$ in  
Medvedev's statement $R=h(L)$, 
are located at the lowest level of an infinite language family hierarchy, indexed by $k=2, 3, \ldots$, called $k$-\emph{strictly locally testable} ($k$-SLT)~\cite{McPa71}. 
Each level $k$ is defined by the set of $k$-factors that may occur in a sentence. The SLT family is the union of all the hierarchy members.
Ii is obvious that the identity stated in MT continues to hold if, instead of a 2-SLT (i.e., local) language, we use a language that is $k$-SLT for $k \geq 2$. 
For clarity we refer to such formulation as the \emph{Extended} Medvedev's Theorem (EMT). 
\par
Four integer parameters of EMT are here relevant: 
\begin{enumerate}
	\item the size of the state set $Q$ of a minimal finite automaton (FA) recognizing $R$;
	\item the cardinalities of the local $\Lambda$ and terminal $\Sigma$ alphabets, or, better, the ratio $\frac{|\Lambda|}{|\Sigma|}$ that we call \emph{alphabetic ratio};
	\item  the value of parameter $k$ that determines the \emph{order} of strict local testability. 
\end{enumerate}
\par
About the local alphabet, Medvedev's statement~\cite{Eilenberg74} (for $k=2$) says that the local alphabet has size\footnote{In the original statement the value was $ |\Sigma| \cdot |Q|^2$ but the lower value can be easily found.} $|\Lambda|= |\Sigma| \cdot |Q|$.
Differently said, MT states that every regular word language of state complexity $|Q|$ is the projection of a 2-SLT language over an alphabet that is $|Q|$ times larger than the terminal alphabet, 
i.e. the alphabetic ratio is $|Q|$ and therefore depends on the language state complexity. 
\par
Much later, the analysis of the alphabetic ratio in the extended MT has obtained the following property~\cite{DBLP:journals/ijfcs/Crespi-ReghizziP12} that is the starting point for the present work.
For every regular language $R$ there is a value $k$ such that $R$ is the projection of a $k$-SLT language over an alphabet of size $2\cdot |\Sigma|$, i.e., with alphabetic ratio 2. 
Moreover, such ratio is minimal. The third parameter, the  order $k$, is in $\mathcal{O}(\log |Q|)$.

\par
The main question we address and solve here is whether an extended Medvedev's theorem holds for \rec~languages and what are the values of the alphabetic ratio and of the order parameter $k$. 
 Concerning the latter, we need to say what in 2D corresponds to the family of $k$-SLT word languages.
Recalling that for $k=2$ the local picture languages are defined by a set of 2 by 2 tiles over $\Lambda$, for any finite $k$ the family of $k$-SLT picture languages is defined by the set of $k$ by $k$ tiles ($k$-tiles) that may occur in a picture. A tiling systems that uses $k$-tiles is called a $k$-tiling system.
\par
We anticipate the main result (Theorem~\ref{thm-main}): for any picture language $R$ in \rec~there exist 
	$k\ge 4$ and a $k$-tiling system with alphabetic ratio 2, recognizing $R$.
Moreover, if $n$ is the size of the local alphabet of a tiling system recognizing $R$, then the value $k$ is $\mathcal{O}(\lg n$).
\par
It may help understanding to fix a specific case; imagine that $R$ is a black and white picture language, then it is the projection of an SLT picture language over a four letters alphabet.
\par 
Next we outline the articulation of the proof, mentioning on one hand how it compares with the EMT proof for REG and, on the other hand, the new lemmas that we have formalized, some of them of independent interest. 
\par
The proof of EMT for REG in~\cite{DBLP:journals/ijfcs/Crespi-ReghizziP12} involves the following concepts. 
\begin{enumerate}
\item We sample with rate $k$ an accepting run of a minimal finite automaton (FA) recognizing language $R$ thus identifying the states at distance $k$ steps.
\item We encode each sampled FA state by means of a binary code-word of length $k$ taken from a comma-free code dictionary. The code bits are then distributed on the following $k$ steps. The property of self-synchronization of comma-free codes permits to avoid erroneous decoding.
\item It is known that the concatenation closure of a comma-free code dictionary is a $2k$-SLT language if the code-word length is $k\ge 2$. This ensures that the pre-image is an SLT language.
\item We handle the final steps of the run when they fall short of a complete sample $k$ and thus they would not match a whole code-word.
\end{enumerate}
 We mention that a similar approach, {using unary rather than } binary codes, was already applied in~\cite{DBLP:journals/jcss/Thomas82} to prove a result on the logical definition of REG. In a more general setting of local functions, another proof of EMT is in~\cite{DBLP:conf/cai/Crespi-Reghizzi19}.
\par
The proof significantly changes moving from words to pictures along the following lines; they may be difficult to understand at first reading but they hopefully convey some useful intuition of the paper technical content.
\begin{enumerate}
	\item  Given a tiling recognizable picture language $R$ over the alphabet $\Sigma$ (i.e. $R$ is the projection of local language $L$ over an alphabet $\Gamma$), instead of sampling an FA run, as in the string case, we tessellate the pre-image in $L$ of a given picture of $R$ into square subpictures of sike $k$, that completely pave the picture, assuming for now that its sides are multiple of k.
	\item In each $k$ by $k$ square we place a binary 2D comma-free code-picture that encodes the periphery of the pre-image square, i.e. a 1-thin square ring of side $k$, to be called a frame. 
	 (the frame information suffices to determine which compositions of such $k$-tiles are elements of $L$). This allows to code the $k$ by $k$ squares obtained by the above tessellation with a comma-free code over the alphabet $\Sigma \times \{0,1\}$.
\par
For the preceding step, we had to design a new family of comma-free 2D codes that has a numerosity suitable for our purposes. 
This family should be a worth addition to the currently limited knowledge of 2D comma-free codes.
\par
	 We prove that the set of pictures tessellated by comma-free codes of size $k$ is a $2k$-SLT languages, and then we derive, by means of the notion of picture morphism, that the language $R$ is projection of a $2k$-SLT language over an alphabet of size $2 |\Sigma|$.
	{Such a result is known for word languages, but its extension to pictures is new.}
	\item The case of pictures with a side not multiple of $k$ is handled by a new padding technique that enlarges the picture.
\end{enumerate}
We stress just one subtle aspect of such an articulated proof. At item (2) we encoded just the frame and not the whole $k$ by $k$ pre-image square. This permits, for sufficiently large $k$, to use the $k^2$ bits of a comma-free code-picture to encode the ring of $4\cdot k$ terminal symbols, i.e., the pre-image frame. A conservative computation of the numerosity of our new comma-free code dictionary shows that any \rec~picture can be adequately encoded.
\medskip
\par
Sections and contents. Section~\ref{s-preliminaries} contains the basic notions of picture languages. Section~\ref{s-kTS} deals with $k$-SLT picture languages and their use in tiling systems. Section~\ref{s-comma-free} introduces the comma-free codes for pictures and a new code family, then proves that the closure of such codes is in SLT. 
Section~\ref{s-main}  proves the minimality of the alphabetic ratio two; then, it introduces the padded languages, whose pictures have sides multiple of a given $k\ge 2$; 
the EMT is then first proved for padded languages; finally, the padding is eliminated, proving the EMT for all REC languages. 
The Conclusion raises a general question about the possibility of similar results for other families of languages different from REG and \rec.

\section{Preliminaries}\label{s-preliminaries}
All the alphabets to be considered are finite.
The following concepts and notations for picture languages follow mostly~\cite{GiammRestivo1997}.
\begin{definition}[picture and picture language]\label{def:pictureLanguage}
A {\em picture} is a rectangular array of letters over an alphabet.
Given a picture $p$, $|p|_{row}$ and $|p|_{col}$ denote the number of rows and columns, respectively; $|p|=\left(|p|_{row},|p|_{col}\right)$ denotes the \emph{picture size}. Two pictures of identical size are called \emph{isometric}.
The set of all pictures over $\Sigma$ of size $(m,n)$ is denoted by $\Sigma^{m, n}$ 
 and the set of all finite pictures over $\Sigma$ is denoted by $\Sigma^{++}$.
\par\noindent
A {\em picture language} over $\Sigma$ is a subset of $\Sigma^{++}$.
\end{definition}
 In the following, the term ``language'' always stands for picture language, and word languages are qualified as such.
\par	
 For a picture $p$ of size $(m,n)$ over an alphabet $\Gamma$ we also use the short notation 
$p_{i,j}\in \Gamma$ to stand for the pixel at position $(i,j)$ in the picture:
\begin{center}
	\scalebox{1}{ 
	$ p = 
		\myLargeTile{p_{1,1}}{\dots}{p_{1,n}}{\dots}{\dots}{\dots}{p_{m,1}}{\dots}{ p_{m,n}}
		$
 	}%scalebox
\end{center}

\paragraph{Concatenations} Let $p, q$ be two pictures. % in $\Sigma^{++}$. 
The {\em horizontal} (or {\em column}) {\em concatenation} $p\obar q$
is intuitively defined when $|p|_{row}= |q|_{row}$ as: 
\def\bBox#1#2{\makebox[#1]{#2}} %vedi p. 104 libro PSTRICKS: scatole di eguale larghezza
\scalebox{0.8}{
\psframebox{\bBox{0.8cm}{$p$}} \psframebox{\bBox{0.8cm}{$q$}}
}%scalebox
. 
The {\em vertical} (or {\em row}) {\em concatenation} $p \ominus q$ is similarly
defined when $|p|_{col}= |q|_{col}$ as:
\scalebox{0.8}{
\begin{tabular}{c}\psframebox{\bBox{0.8cm}{$p$}}\\ \psframebox{\bBox{0.8cm}{$q$}}
\end{tabular}
}%scalebo
.
Concatenations are extended to languages in the obvious way.
 We also need the power of the two concatenations, respectively denoted by $p^{\ominus h}$ and $p^{\obar h}$, $h\geq 1$.
\par
The notations $\Sigma^{m, n}$ and $\Sigma^{++}$ are immediately extended by considering, instead of an alphabet $\Sigma$, a finite set $F \subseteq \Sigma^{k, k}$, $k \geq 1$ of (isometric) pictures.\footnote{In the literature the notation $F^{++}$ is defined for any set $F$ of pictures, but we do not need such greater generality.} 
$F^{++}$ denotes the closure under horizontal and vertical concatenations of the pictures in $F$.
We also need the closure under horizontal concatenation and the closure under vertical concatenation, denoted respectively as $F^{\obar +}$ and $F^{\ominus +}$.
\par
Since the symbols on the boundary of picture often play a special role for recognition, it is convenient to surround them by a strip of width one, called \emph{frame}, of the reserved symbol $\sharp$ that may not occur within a picture. A picture $p$ with its frame is called \emph{bordered} and denoted by $\widehat{p}$; 
it has size $(|p|_{row}+2)$, $(|p|_{col}+2)$. We extend the notation to a language by writing $\widehat L= \{\widehat p \mid p \in L\}$.
\par
A \emph{subpicture of $p$}, denoted by $ p_{\subd i j {i'}{j'}}$, is the portion of $p$ defined by the top-left coordinates $(i,j)$ and by the bottom right coordinates
$({i'}, {j'})$, with $1\le i \le i'\le |p|_{row}$, and $1\le j \le j'\le |p|_{col}$.
The set of all subpictures of size $(2,2)$ (if any) of a picture $p$, called {\em tiles}, is denoted as $\ktile{p}{2}$.

\paragraph{Tiling recognition}\label{s-tiling}
We recall the classical definition of tiling recognizable language as the image under an alphabetic projection of a local language.
\par
A language $L\subseteq \Sigma^{++}$ is \emph{local} if there exists a finite set $\tset_2$ of {\em tiles } 
	in $(\Sigma\cup \{\#\})^{2,2}$ such that 
	$L=\{p\in \Sigma^{++} \mid \ktile{\widehat p}{2}\subseteq \tset_2 \}$; 
	we also write $L = L(\tset_2)$.
\par
Let $\Gamma$ and $\Sigma$ be alphabets.
 Given a mapping $\pi:\Gamma\rightarrow \Sigma$, to be termed {\em projection}, we extend it to pairs of isometric pictures $p'\in \Gamma^{++}$, $p \in \Sigma^{++}$
by: 
\[
p= \pi(p') \text{ such that } p_{i,j}=\pi(p'_{i,j}) \text{ for all }(i,j) \in \, 1\ldots |p'|_{row} \, \times\, 1\ldots |p'|_{col}.
\] 
Then,
$p'$ is called a \emph{pre-image} of $p$. 
\begin{definition}[tiling system]
	\label{def:TS}\label{def:REC}
	A {\em tiling system }(TS) is a quadruple $(\Sigma,\Gamma,\tset,\pi)$ where $\Sigma$ and $\Gamma$ are respectively the \emph{terminal} and the \emph{local} alphabets, 
	$\tset\subseteq(\Sigma\cup \{\#\})^{2,2}$ is the {\em tile set} and $\pi:\Gamma\rightarrow \Sigma$ is a projection.
	A language $L\subseteq \Sigma^{++}$ is recognized by such a TS 
	if 	$L=\pi(L(\tset))$. The family of all {\em tiling recognizable} languages is denoted by \rec.
\end{definition}
It is worth observing that the above definition includes as a special case the family of regular word languages. It suffices to view a word $x\in \Sigma^+$ of length $n$ as an one-row picture $x'$ of size $(1,n)$. The bordered version of $x'$ is
\[
\widehat x' = \sharp^{n+2} \ominus \left(\sharp \obar x' \obar \sharp\right)\ominus\sharp^{n+2}\,.
\] 
Then, the definition of a \rec~language $L'$ of one-row pictures as $L'=\pi(L'(\tset))$ immediately corresponds to the definition of the regular word language $L$ as the projection of the local word language defined by the rectangular ``tiles'' 
of size $(1,2)$ occurring in the set $\tset$. Such a way of definining a regular word language corresponds to the historical definition known as Medvedev's theorem~\cite{Medvedev1964,Eilenberg74}. 
Notice that in the case of word languages, a pre-image of a word $w$ in the tiling system matches  the notion of an accepting run of $w$ using a finite automaton recognizing the same language.

\section{Tradeoff between alphabet cardinality and tile size}\label{s-kTS}
In this section we consider the role and interdependence of two basic parameters present in a TS, the local alphabet cardinality and the size of the tiles. The latter was fixed to $(2,2)$ in Definition~\ref{def:REC}, but here we allow larger tiles. 
The hierarchies of language families induced by the two parameters are stated. We finish with an example showing how tile enlargement may permit to reduce the local alphabet cardinality. 

\paragraph{Local alphabet cardinality}
We first formalize the folklore fact that the local alphabet size needed to define a \rec~language is a measure of its complexity, so that such a parameter induces an infinite hierarchy on the \rec~family. 
\par
 Let $\ell\ge 1$ be the cardinality of the local alphabet $\Gamma$ in Definition~\ref{def:REC}.
\begin{proposition}\label{prop-recHierarchy}
For every $\ell\ge 1$, let $\rec_\ell$ be the family of languages recognized by tiling systems with a local alphabet of cardinality at most $\ell$. Then, the following inclusion holds, 
$\rec_\ell \subsetneq \rec_{\ell+1}$.
\end{proposition} 
\begin{proof} 
Let $\ell\ge 1$ and consider a TS $(\{a, b\}, \Gamma, \tset,  \pi)$ accepting the word language $R_\ell=\{a^{\ell -1} b\}^+$. 
We claim that $|\Gamma|\ge \ell$. 
\par\noindent
By contradiction, assume that there is a TS recognizing $R_\ell$ such that 
$\Gamma = \{1, 2, \dots, j\}$ for some $j<\ell$. Let $\alpha = i_1 i_2 \dots i_j, \dots i_\ell \in L(\tset)$ (whose projection $\pi$ is $a^{\ell -1} b$), for suitable $i_1, i_2, \dots, i_j, \dots, i_\ell \in \Gamma$. 
Therefore, the tiles of the form $\mytileDynamic{\#}{\#}{i_h}{i_{h+1}}$, for all $1\le h<\ell$ must be in $\tset $. 
Since $j<\ell$, there exist $m,n$, with $1\le m<n\le \ell$ such that $\alpha= i_1 i_2 \dots i_m i_{m+1} \dots i_{n-1} i_n \dots i_\ell$ with $i_m = i_n$.
\par\noindent Therefore, the picture $\beta= i_1 i_2 \dots i_m i_{m+1} \dots i_n i_{m+1} \dots i_{n-1} i_n \dots i_\ell$ has the same tiles of $\alpha$,
hence also $\beta \in L(T)$,
with $\pi(\beta) = a^{\ell-1+m-n} b$, a contradiction. 
\qed
\end{proof}

\paragraph{Tiles and strict local testability}
 We consider the second parameter of interest to us, namely the tile size.
We lift from one to two dimensions a well-known approach for defining word languages, in order to introduce a language family characterized by the subpictures of size $k$. 
\newcommand{\dhat}[1]{{\widehat{\widehat{ #1\,}}}}
\par
Given $k \geq 2$, we denote by $\ktile{p}{k}$ the set of all subpictures of size $(k,k)$, if any, that occur in picture $p$. 
The set of all subpictures of size $(k,k)$ of a language $L$ is defined as $\ktile{L}{k}= \bigcup_{p \in L}\ktile{ p}{k}$.
	The elements of the sets $\ktile{p}{k}$ and $\ktile{L}{k}$ are called the {\em $k$-tiles} of $p$ and respectively of $L$. 

In the definition of local languages, the membership of a picture $p$ in the language is determined by the set $\ktile{\widehat p}{2}$,
i.e., the set of subpictures of size $(2,2)$ of $\widehat p$. 
Even in the case $p$ is composed by a single letter, $\ktile{\widehat p}{2}$ is well defined. 
The subpictures in $\ktile{\widehat p}{k}$ are, however, well defined only if $p$ has size at least $(k-1,k-1)$. 
To obviate the issue, we enlarge the border containing the reserved symbol $\#$.

\par
For any $k \ge2$, for any picture $p \in \Sigma^{++}$, we denote by $\dhat{\, p}$ the picture with thicker border, obtained by surrounding $\widehat p$ with the minimum number of rows of \#’s at the bottom of $p$ and with the minimum number of columns of \#’s at the right of $p$
such that the subpictures of $\dhat{\, p}$ of size $(k,k)$ are defined. 
Remark that, if a picture $p$ has size $(m,n)$ with $m,n \ge k-1$, then the picture with thicker border $\dhat{\, p} = \widehat p$. 
In particular, if $k = 2$, for every picture $p$, $\dhat{\, p} = \widehat p$.

	\begin{definition}[strict local testability]\label{def:SLT}
		Given $k\ge 2$, a language $L\subseteq \Sigma^{++}$ is {\em $k$-strictly-locally-testable} ($k$-SLT) if there exists a finite set 
		$\tset_k\subseteq \left(\Sigma\cup\{\#\}\right)^{k,k}$
		such that 
		$L=\left\{p\in \Sigma^{++} \mid \ktile{\dhat{p\ }}{k}\subseteq \tset_k \right\}$; 
		we also write $L = L(\tset_k)$. The value $k$ is called the order of $L$. A language is called strictly-locally-testable (SLT) if it is $k$-SLT for some $k \ge 2$. 
	\end{definition}
	In other words, to check that a picture $p$ belongs to a $k$-SLT language $L$, we check that each subpicture
	of size $(k, k)$ of the picture  with thicker border $\dhat{p}$ is included in a given set of $k$-tiles. 
	In particular, a local language is the special case $k=2$, i.e., a 2-SLT language.
	\par

\par
Since $k$-SLT (picture) languages include as a special case $k$-SLT word languages, the following proposition derives immediately from a known language hierarchy (e.g. in~\cite{McPa71}).
\begin{proposition}\label{prop:k-hierarchy}
For every $k \geq 2$, the family of $k$-SLT languages over a terminal alphabet of cardinality $|\Sigma|>1$ is strictly included in the family of $(k+1)$-SLT languages. 
\end{proposition}

\par
It is quite natural to 
 generalize Definition~\ref{def:TS} from a system of 2-tiles to the case of larger tiles.
\begin{definition}[$k$-tiling system]
	\label{def:k-TS}
Let $k\geq 2$ be the tile size. 	A {\em $k$-tiling system }($k$-TS) is a quadruple $(\Sigma,\Gamma,T_k,\pi)$ where the alphabets $\Sigma$, $\Gamma$ and the projection $\pi$ are as in Definition~\ref{def:TS}, and
	$\tset _k\subseteq(\Gamma \cup \{\#\})^{k,k}$.
	A language $L\subseteq \Sigma^{++}$ is recognized by such a $k$-TS 
	if 	$L=\pi(L(\tset _k))$. 
		\par\noindent 	The \emph{alphabetic ratio} of a $k$-TS is defined as the quotient $\frac{|\Gamma|}{|\Sigma|}$.
\end{definition}
Note on terminology: 
we keep using the terms pre-image and local alphabet as in Definition~\ref{def:TS}.
It is worth observing that, for a given $k\ge 2$, the alphabetic ratio may be considered a measure of the complexity of a $k$-tiling system.
\par
From the preceding definition and from Proposition~\ref{prop:k-hierarchy}, a natural question arises: whether the family of $k$-recognizable languages strictly includes \rec. The known answer is negative and follows from the next property. If we apply a projection to $k$-SLT languages, the hierarchy of Proposition~\ref{prop:k-hierarchy} collapses; this result is proved in~\cite{GiammarresiRestivo1992,DBLP:journals/iandc/GiammarresiRST96}, and we restate it to prepare the concepts needed later.

\begin{theorem}\label{thm-klocal}
Given a $k$-SLT language $L\subseteq \Sigma^{++}$ defined as $L=L(\tset_k)$ (where $\tset_k$ is a set of $k$-tiles), there exists an alphabet $\Gamma$, a local language $L'\subseteq\Gamma^{++}$ and a projection $\pi: \Gamma\to\Sigma$ such that $L=\pi(L')$.
\end{theorem}	

\begin{remark}\label{rem-sizeGamma}
Both proofs in~\cite{GiammarresiRestivo1992,DBLP:journals/iandc/GiammarresiRST96} consider a local alphabet $\Gamma$ of size $|\Gamma|=|\Sigma|\cdot |\tset_k|$.
Since $\tset_k$ is a subset of $(\Sigma\cup \{\#\})^{k,k}$, one has $|\tset_k|\le (|\Sigma|+1)^{k^2}$ and 
 $|\Gamma| \le (|\Sigma|+1)^{k^2 + 1}$.
\end{remark}

It follows that the family of SLT languages is strictly included in \rec~\cite{GiammRestivo1997} and that the use of larger tiles does not enlarge the \rec~family.
\begin{corollary}\label{cor-kAsTwoTiles}
The family of $k$-TS recognizable languages coincides with the family \rec~of TS recognizable languages. 
\end{corollary}
\par\medskip
\paragraph{Role of the local alphabet size}
We have seen that any \rec~language over $\Sigma$ can be obtained both as a projection of a local language over alphabet $\Gamma_2$, 
and as a projection of a $k$-SLT language (with $k>2$) over alphabet $\Gamma_k$.
However, if we use 2-tiles instead of $k$-tiles, 
{then it often happens that the alphabet $\Gamma_2$ is larger than $\Gamma_k$.}
 In other words, $k$-tiling systems typically exhibit a trade-off between the two parameters: tile size and local alphabet size. 
\par
The next example illustrates Corollary~\ref{cor-kAsTwoTiles}, Proposition~\ref{prop-recHierarchy} and the trade-off between the two parameters.

\newcolumntype{C}{>{$}p{10pt}<{$}} % centered, in math mode, fixed width

\begin{figure}[tbh]\label{fig-double-squares}
\begin{center}
\setlength{\tabcolsep}{10pt}
\renewcommand{\arraystretch}{1.5}
\scalebox{0.8}{
\begin{tabular}{cc}
\emph{Pre-image $r_3$ with 3 symbols} &\emph{ Pre-image $r_2$ with 2 symbols} % & \emph{Pre-image 3}
\\
$
\begin{array}{|C|C|C|C|C|C|C|C|}%7 col, 8 rows
\hline
 \cellcolor{yellow}\ssearrow&\cellcolor{yellow}n&\cellcolor{white}n&\cellcolor{yellow}n&\cellcolor{yellow}n&\cellcolor{white}n&\cellcolor{yellow}n&\cellcolor{yellow}\nnearrow \\\hline
 \cellcolor{yellow}n&\cellcolor{yellow}\ssearrow&n&\cellcolor{yellow}n&\cellcolor{yellow}n&n&\cellcolor{yellow}\nnearrow&\cellcolor{yellow}b \\\hline
 n&n&\ssearrow&\cellcolor{cyan}n&\cellcolor{cyan}n&\nnearrow&n&n \\\hline
 n&n&n&\cellcolor{cyan}\ssearrow&\cellcolor{cyan}\nnearrow&n&n&n \\\hline
\end{array}
$

& %pre-image 2
$
\begin{array}{|C|C|C|C|C|C|C|C|}%7 col, 8 rows
\hline
 \rightarrow&n&n&n&n&n&n&\rightarrow \\\hline
 n&\rightarrow&n&n&n&n&\rightarrow&n \\\hline
 n&n&\rightarrow&n&n&\rightarrow&n&n \\\hline
 n&n&n&\rightarrow&\rightarrow&n&n&n \\\hline
\end{array}
$
\end{tabular}
}
\end{center}
\hspace{0.5cm} 
\begin{center} 
\setlength{\tabcolsep}{10pt}
\renewcommand{\arraystretch}{1.5}

\scalebox{0.8}{\emph{Pre-image $x$ of illegal picture}}\\
\scalebox{0.8}{$
\begin{array}{|C|C|C|C|C|C|C|C|C|C|C|}%
\hline
 \rightarrow&n&n&n&n&n&n&n&n&\rightarrow \\\hline
 n&\rightarrow&n&\cellcolor{pink}n&\cellcolor{pink}n&\cellcolor{pink}n&n&n&\rightarrow&n \\\hline
 n&n&\rightarrow&\cellcolor{pink}n&\cellcolor{pink}n&\cellcolor{pink}n&n&\rightarrow&n&n \\\hline
 n&n&n&\cellcolor{pink}\rightarrow&\cellcolor{pink}\rightarrow&\cellcolor{pink}\rightarrow&\rightarrow&n&n&n \\\hline
\end{array}
$
}
\end{center}
\caption{Top, left. Pre-image $r_3$ over $\Gamma_3$ of the picture $a^{4,8}$ in language $R$ of Example~\ref{ex:k-localLanguages}; one instance of the 2-tiles of $\ktile{r_3}{2}$ is highlighted. 
Top, right. Pre-image $r_2$ over $\Gamma_2$ of the picture $a^{4,8}$. 
Bottom. Pre-image $x$ over $\Gamma_2$ of the illegal picture $a^{4,10}$: 
the highlighted 3-tile is not present in $\ktile{r_2}{3}$, but all tiles in $\ktile{\widehat{x}}{2}$ are included in $\ktile{\widehat{r_2}}{2}$.}
\label{fig:example}
\end{figure}

\begin{example}	\label{ex:k-localLanguages}
As a simple running example, we take the unary language $R\subseteq \{a\}^{++}$ such that for any $p\in R$, $|p|_{col}=2\cdot|p|_{row}$. We show its definition by means of two tiling systems, first using 2-tiles and then using 3-tiles on a smaller alphabet.
\par\noindent
 The first TS comprises the 2-tiles $\tset _2\subseteq {\Gamma_3}^{2,2}$ with $\Gamma_3= \{n,\ssearrow,\nnearrow\}$ that are visible in the pre-image $r_3$ in Figure~\ref{fig:example}, top left, plus the 2-tiles (not shown) coming from the same picture but bordered, $\widehat{r_3}$, which contain $\sharp$ symbols.
The obvious projection is: for all $c \in \Gamma_3, \pi(c)=a$.
\par\noindent
Next, we merge together the symbols $\ssearrow$ and $\nnearrow$ into the symbol $\rightarrow$ obtaining the local alphabet $\Gamma_2 = \{n,\rightarrow\}$, with the projection $\forall c \in \Gamma_2, \pi(c)=a$. The corresponding pre-image, $r_2$, of the same picture $a^{4,8}$ is shown in Figure~\ref{fig:example}, top right; let $\tset '_2$ be the tiles of $\widehat{r_2}$. But now also the illegal picture $a^{4,10}$ having the pre-image $x$ shown in Figure~\ref{fig:example}, bottom, would be tiled using a subset of $\tset '_2$, hence $\pi(L(\tset '_2))\supset R$. Therefore, 
 the TS $(\{a\},\Gamma_2,\tset'_2,\pi)$ fails to define language $R$.
\par\noindent
To exclude such spurious pictures from the language, still using the same local alphabet $\Gamma_2$, we need larger tiles. We leave to the reader to check that the 3-TS based on the 3-tiles $\ktile{\widehat{r_2}}{3}$ (Figure~\ref{fig:example}, top right) correctly defines language $R$. 
\end{example}

\section{{Comma-free} picture codes and local testability}\label{s-comma-free}
We introduce in this section some notions and results that are used in the proof of our main result in Section~\ref{s-main}.
\par
Given a tiling system $(\Sigma, \Gamma, \tset, \pi)$ recognizing a picture language $R \subseteq \Sigma^{++}$, we consider the local language $L(\tset )$ over the alphabet $\Gamma$.
For a given integer $k \geq 2$, we first reduce the main problem to the case in which the pictures of $R$ (and then also of $L(\tset )$) have size multiple of $k$, i.e., $L(\tset ) \subset \left(\Gamma^{k,k}\right)^{++}$. 
Then the pictures of $L(\tset )$ can be decomposed (tessellated) into subpictures of size $k \times k$. 
This set of subpictures is denoted by $\kblock{L(\tset )}{k}$ and can be considered from the information theory perspective as a two-dimensional code, to be called a \emph{picture code}. The definition of  $\kblock{p}{k}$  follows for a picture $p$ such that $|p|_{row}\Mod k =0$ and $|p|_{col}\Mod k =0$ :
\begin{equation}
\kblock{p}{k} = \left\{p_{\subd i j {i+k\,}{j+k}} \mid i\Mod k = 1,\, j\Mod k =1  \right\}.
\label{eq-defP_kk}
\end{equation}
The definition is naturally extended to a set $L$ of pictures of suitable dimensions. Clearly $\kblock{p}{k}\subset \ktile{p}{k}$.
\par
We proceed to introduce codes for pictures by means of a morphism, then we define the type of codes, called comma-free, needed in later proofs. 
Loosely speaking, the essential property of a comma-free code is that, in a picture tessellated by codes, any subpicture of size $k$ that occurs in a position misaligned with the $k \times k$ grid, cannot be a code-picture. 
\par
Comma-free codes are a classical topic for words, but are less studied for pictures. 
For the latter we introduce a novel family of binary comma-free codes and we compute their numerosity as a function of size. 
The evaluation of such a numerosity, in Proposition~\ref{prop-commafree-family}, is essential in the proof of the main result. 
Our definition of comma-free picture codes is based on square pictures of fixed size; variable-size comma-free picture codes are defined in~\cite{ANSELMO20174}, but finding bounds on the numerosity in that case is still an open problem.
\par
The section finishes with two statements: 
Proposition~\ref{prop-X-k-local} states that the set of pictures tessellated by comma-free codes of size $k$ is strictly locally testable (precisely $2k$-SLT); 
Theorem~\ref{th-SLTCode} states that the image of a local language under a one-to-one morphism, mapping each symbol to a comma-free code-picture of size $k$, is $2k$-SLT. 
These statements, that extend to two dimensions similar properties proved in the case of words, have an independent interest.
 \par
Let us outline how the properties mentioned will permit in Section 5  to build a $2k$-tiling system, 
recognizing the picture language $R \subseteq \Sigma^{++}$, using a local alphabet of size $2\cdot |\Sigma|$.
\par
Given the set $\kblock{L(\tset )}{k}$, defined at Eq.~\eqref{eq-defP_kk}  above,
we associate a pair $(f(r), \pi(r))$ to each $r \in \kblock{L(\tset )}{k}$, where $f(r)$ is the ``frame'' of $r$, and $\pi(r)$ is the projection of $r$ on the terminal alphabet $\Sigma$. 
We denote by $B_k$ the set of such pairs when $r$ runs through the set $\kblock{L(\tset )}{k}$. 
We then define a new tiling system recognizing $R$, having $B_k$ as a local alphabet (\emph{vi}. Lemma~\ref{lemma-main}).
\par
By selecting a sufficiently large $k$, 
the numerosity of a comma-free code of size $k$ in our family is greater 
than the cardinality of the set of ``frames'' of elements in $\kblock{L(\tset )}{k}$ 
So we can encode such frames with a binary comma-free code of size $k$. 
Then, one derives that the symbols of $B_k$ can be encoded with a comma-free code $Z$ of size $k$ over the alphabet $\Sigma \times \{0,1\}$. This allows to define, by using Theorem~\ref{th-SLTCode}, a $2k$-tiling system over an alphabet of size $2\cdot|\Sigma|$ recognizing $R$.

\subsection{Picture morphisms and picture codes}\label{sect-pict.morphismsAndCodes}
 \begin{definition}[picture morphism]
 	\label{def:pictureMorphism}
 	Given two alphabets $\Gamma, \Lambda$, 
 		a (\emph{picture}) \emph{morphism} is a mapping
 $ 	\varphi: \Gamma^{++} \to \Lambda^{++}$ such that, for all 
 $p, q \in \Gamma^{++}$ : 
 
 \[\left\{\begin{array}{ll}
 	i)&\varphi(p \obar q)=\varphi(p) \obar \varphi(q)
 	\\
 	ii)&\varphi(p \ominus q)=\varphi(p) \ominus \varphi(q)
 	\end{array}
 	\right. 
 	\] 
	\end{definition}
 	This implies that the images by $\varphi$ of the elements of alphabet $\Gamma$ are isometric, 
 	i.e., for any $x,y \in \Gamma$, $|\varphi(x)|_{row}=|\varphi(y)|_{row}$ and $|\varphi(x)|_{col}=|\varphi(y)|_{col}$. 
\par
Notice that, unlike the case of words, a picture morphism $\varphi: \Gamma^{++} \to \Lambda^{++}$ is one-to-one if the restriction to $\Gamma$, namely the mapping $\varphi: \Gamma \to \Lambda^{++}$, is one-to-one.

\paragraph{Code-words}
It may help to recall the basic notion of uniform (i.e., fixed-length) code for words. Given two alphabets $\Gamma, \Lambda$ and a one-to-one morphism
 	$\varphi: \Gamma^{+} \to \Lambda^{+}$, such that for all $x\in \Gamma$ the image $\varphi(x)$ is a word of length $k\geq 2$, the set $X=\varphi(\Gamma)$ is a code, and each of its elements is a code-word. 
 	It follows that any word in $X^+$ admits exactly one encoding into code-words.

	\par
	We also recall the definition of comma-free code~\cite{GolombEtAl58,BerPerReu09}. 
	A code $X\subseteq \Lambda^k$, $k\ge 1$, is \emph{comma-free}\footnote{The term ``comma-free'' suggests that such codes do not use a reserved character 
	(the ``comma'') or a reserved substring to mark the separation between code-words.} if $X^2\cap y X z = \emptyset$ for all $y, z \in \Lambda^+$.

\par
In our research, the use of comma-free codes is motivated by their preserving the SLT property~\cite{HashiguchiHonda1976,Restivo1974}, i.e,. for any such comma-free code $X=\varphi(\Gamma)$, the language $X^+$ 
is SLT and if $L$ is a SLT language
 over the alphabet $\Gamma$, 
 then $\varphi(L)$ is also SLT. 
\par
For instance, given the comma-free code of length 5 $Y = \{00111, 00001, 10001\}$, it is immediate to notice that $Y^+$ is $10$-SLT. 
We are going to prove that also in two dimensions the comma-free picture codes anlogous properties.

\paragraph{Picture codes} 
We proceed to define the picture codes and the comma-free ones, finishing with their strict local testability property.
 \begin{definition}[picture code]
 	\label{def:pictureCode}
 	Given two alphabets $\Gamma, \Lambda$ and a one-to-one morphism
 	$\varphi: \Gamma^{++} \to \Lambda^{++}$, 
 	the set $X= \varphi(\Gamma) \subseteq \Lambda^{++}$
		is called a (uniform) \emph{picture code}; its elements are called \emph{code-pictures}. 
		 	For convenience, the morphism ``$\varphi$'' will be also denoted with $\codeX{-}:\Gamma^{++}\to \Lambda^{++}$. 
		\par\noindent
 	For $\gamma\in \Gamma^{++}$, 
 	the picture $\codeX{\gamma}\in \Lambda^{++}$ is called the {\em encoding} of $\gamma$. 
 	
 \end{definition}

The set $X^{++}$ is defined as $\varphi\left(\Gamma^{++}\right)$, i.e., the set of all pictures over the alphabet $\Lambda$ 
defined as (horizontal/vertical) concatenations of the code-pictures of $X$. 

\paragraph{Tessellation}
 A useful concept when dealing with encodings is the tessellation (or "paving"). 
	Given a value $k>0$ and a picture $p \in \left(\Gamma^{k,k}\right)^{++}$, let the $k$-\emph{tessellation} 
	be defined as the unique decomposition of $p$ into square subpictures of size $k \times k$.
\par\noindent
We recall the notation ${\kblock{p}{k}}$ (in Eq.~\eqref{eq-defP_kk} above) to define the set of all pictures in the $k$-tessellation of $p$. 
The notation is extended to a language $L$ of pictures with both rows and columns multiple of $k$ as: 
${\kblock{L}{k}} = \left\{r \in {\kblock{q}{k}} \mid q \in L\right\}$. Remark that $L \subseteq ({\kblock{L}{k}})^{++}$.
\begin{figure}
\setlength{\tabcolsep}{10pt}
\renewcommand{\arraystretch}{1.0}
\begin{center}
\scalebox{0.7}{
$p= 
\begin{array}{|ccc: ccc: ccc: ccc: ccc|} \hline
 a &a &a &a &a &a &a &a &a &a &a &a &a &a &\$\\
 a &a & a &a &a &a &a&a &a&a & a &a & a &a &\$
\\
 a &a &a & a &a &a &a& a &a&a &a &a &a &a &\$
\\ \hdashline
 a &a &a &a& a &a &a& a &a &a &a &a& a &a & \$
\\
 a &a &a &a &a& a &a &a&a &a &a &a &a & a &\$
\\
 a &a &a &a &a &a&a &a&a &a &a &a &a & a &\$
\\ \hdashline
 a &a &a &a &a &a &a&a &a &a &a &a & a & a &\$
\\ \$ & \$& \$ & \$& \$ &\$& \$&\$& \$ & \$ & \$ &\$ & \$ & \$ &\$
\\ \$ & \$ & \$ & \$ & \$ &\$& \$&\$& \$ & \$ & \$ &\$& \$ & \$ &\$
\\
\hline
\end{array}
$
}
\\ \vspace{0.5cm}
\scalebox{0.7}{
$
\kblock{p}{3}= \left\{ \begin{array}{|ccc|} \hline
				a & a & a \\
				a & a&a \\
				a & a & a \\ \hline
			\end{array}, 
			\begin{array}{|ccc|} \hline
				a & a & a \\
				\$ & \$& \$ \\
				\$ & \$& \$ \\ \hline
			\end{array}, 
			\begin{array}{|ccc|} \hline
				a & a & \$ \\
				a & a&\$ \\
				a & a & \$ \\ \hline
			\end{array},
			 \begin{array}{|ccc|} \hline
				a & a & \$ \\
				\$ & \$& \$ \\
				\$ & \$& \$ \\ \hline
			\end{array}\right\}
$
}
\end{center}
\caption{The 3-tesselation of a picture $p$ and the subpictures of its $3$-tesselation.}\label{fig-tesselation} 
\end{figure}

An immediate consequence of the definition is that if a picture $p \in \Lambda^{++}$ is the encoding of a picture $\gamma$, i.e., $p=\codeX{\gamma}$, then 
the $k$-tessellation of $p$ exclusively includes  as subpictures the code-pictures of  $X$, i.e.,
$\kblock{p}{k} \subseteq X$. 

 %For instance, ${\kblock{p_0}{3}}$ for the picture $p_0$ of Figure~\ref{fig-local-padding} includes four 3-tiles, shown in Figure~\ref{fig-3tiles}.
For instance, in Figure~\ref{figCommaFreeFamily}, left part, of Section~\ref{sect-newFamilyCodes}, the picture $p$ has a two-by-two tessellation with $k=5$; 
the set ${\kblock{p}{5}}$ for picture $p$ includes just three 5-tiles, since those at positions n.e. and s.w. are identical.
\par\noindent

\

\par
We are ready to generalize the notion of comma-free code from words to pictures. 
Let $p$ be a picture of size $(r,c)$; a subpicture $p_{(i,j;\,n,m)}$, such that 
$1<i\leq n<r$ and $1<j\leq m<c$ is called \emph{internal}.
\par
 Given a set $X\subseteq \Lambda^{k,k}$, consider $X^{2,2}$, i.e., the set
 of all pictures $p$ of size $(2k,2k)$ of the form $(X\obar X)\ominus (X\obar X)$.
\begin{definition}[comma-free picture code]\label{def-commafree}
Let $\Lambda$ be an alphabet and let $k\ge 2$. A picture code $X\subseteq \Lambda^{k,k}$ is 
{\em comma-free } if, for all pictures $p \in X^{2,2}$, there is no internal subpicture $q\in\Lambda^{k,k}$ of $p$ such that $q \in X$. 
\end{definition}
It should be clear that the above is the natural transposition in two dimensions of the classical concept of comma-free code. A schematic example is in Figure~\ref{fig-CommaFreeDef}.

\begin{figure}
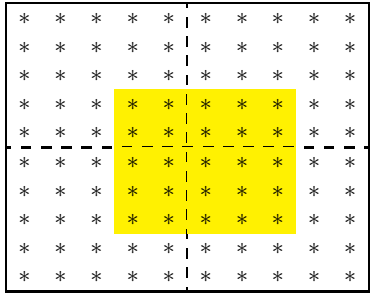
\begin{center}\scalebox{0.9}{
$
\begin{array}{|lllll:}
			\hline
			*&*&*&*&* \\
			 *& *& *& *& * \\
			*&*&*&*&*\\
						\cline{4-5}
			 *& *& * &\cellcolor{yellow} *&\cellcolor{yellow} *\\
			 *& *& *&\cellcolor{yellow} *&\cellcolor{yellow} *\\
						\hdashline
			*&*&*&\cellcolor{yellow}*&\cellcolor{yellow}* \\
			 *& *& *&\cellcolor{yellow} *&\cellcolor{yellow} *\\
			*&*&*&\cellcolor{yellow}*&\cellcolor{yellow}*\\
			 *& *& *& *& *\\
			 *& *& *& *& *\\
			\hline
		\end{array}
\begin{array}{lllll|}
						\hline
			*&*&*&*&*  \\
			 *& *& *& *& * \\
			*&*&*&*&* \\
			\cellcolor{yellow} *& \cellcolor{yellow} *&\cellcolor{yellow} *& *& *\\
			\cellcolor{yellow} *&\cellcolor{yellow} *& \cellcolor{yellow} *& *& * \\
			\hdashline 
			\cellcolor{yellow}*&\cellcolor{yellow}*&\cellcolor{yellow}*&*&*  \\
			\cellcolor{yellow} *&\cellcolor{yellow} *&\cellcolor{yellow} *& *& * \\
			\cellcolor{yellow}*&\cellcolor{yellow}*&\cellcolor{yellow}* &*&*  \\
			 *& *& *& *& *\\
			 *& *& *& *& *
			\\ \hline
		\end{array}
		$
}		
\end{center}
\caption{A picture in $X^{2,2}$, where $X$ is a picture code. An asterisk stands for any symbol. 
The picture highlights a generic internal position where the presence of a code-picture in $X$ is forbidden if $X$ is comma-free.}\label{fig-CommaFreeDef}
\end{figure}		
\subsection{A new family of comma-free picture codes}\label{sect-newFamilyCodes}
 Very few (if any) examples of comma-free picture codes are available
 such that their \emph{numerosity} (meaning the number of code-pictures) is known. 
 An upper bound on the numerosity of a comma-free picture code can immediately be derived by Theorem 12 of~\cite{GAMARD201758}, which computes the number of primitive pictures of a given size, while a lower bound can be derived by 
 considering the family of non-overlapping picture codes~\cite{DBLP:journals/tcs/AnselmoGM17}, a bound we used in the initial version of this research~\cite{Crespi-ReghizziRestivoSanPietro21}. 
 \par
 We present here a new family of comma-free picture codes that exploits the property of comma-free code-words on carefully selected rows and columns; 
 we also compute a lower bound on its numerosity, greater than in the case of non-overlapping picture codes.
 Such a family may be of some interest of its own apart from its instrumental use in later proofs.

\par
We need a few simple definitions.
Let $w=w_1 \ldots w_n \in \Lambda^+$, with each $w_i \in \Lambda$, be a word; let $col(w)$ denote the picture $col(w)=w_1 \ominus \ldots \ominus w_n$; the notation is naturally extended to set of words.
The {\em $i$-left-rotation} of $w$
is the word $w_{i+1} w_{i+2} \dots w_n w_1 \dots w_i$.
 An $i$-left-rotation is thus a circular permutation. 
\begin{definition}[Obligation word] An \emph{obligation word} is a non-empty Boolean string over $w \in \{t, f \}^+$ such that 
for every circular permutation $\widetilde w$ of $w$ the bit-by-bit logical product $w \,\wedge \,\widetilde w \neq f^k$. 

\end{definition}

Thus, if $w=w_1 \dots w_n$ is an obligation word and $\widetilde w$ is one of its circular permutations, then there is at least 
one position $i$, $1 \leq i \leq n$, such that $w_i = \widetilde{w}_i = t$. For example, $tffff$ is
 not an obbligation word since its 1-left-rotation $fffft$ is such that $tffff \wedge fffft =fffff$; on the other hand, 
 it is easy to check that $w=ftftt$ is an obligation word since for every circular permutation $\widetilde{ftftt}$, $\widetilde{ftftt} \wedge ftftt \neq fffff$.

\begin{definition}[Family of comma-free picture codes]\label{def-family-comma-free}
Let $Y_{hor}$ and $Y_{vert}\subseteq \Lambda^k$ be comma-free word codes of length $k\ge 3$, respectively referred to as \emph{horizontal} and \emph{vertical}.
Let $w \in \{t, f \}^k$ be an obligation word.
\par\noindent
 Define the finite, uniform, language substitution 
 $\sigma: \{t, f \}^+ \to 2^{\Lambda^+}$ by means of 
\[
\sigma(t)=Y_{hor} \quad \text{ and } \quad \sigma(f)=\Lambda^k
\]
($\sigma$  can obviously operate also on a one-column picture.)
\par\noindent
The code $X\subseteq \Lambda^{k,k}$ is the set of pictures in $\Lambda^{k,k}$ meeting the condition: 
\begin{equation}
\sigma\left(col(w)\right) \, \cap \, \left(col(Y_{vert}) \obar  \Lambda^{k, k-1}\right)
\end{equation}
\end{definition}
Less formally, given an obligation word $w$ located in a one-column picture $col(w)\in \{t, f \}^{k\ominus}$, a comma-free code-picture $x$ 
has a vertical code-word in the first column and a horizontal code-word in every row $1\leq i \leq k$ where the presence of a code is obligated by the occurrence of $t$ in the position $i$ of the obligation word. 
\par
In later use, it happens  that $Y_{hor}=Y_{vert}$, i.e, the vertical comma-free code used for a column, is the same horizontal code used for the rows --in this case, both are referred to as $Y$. 

\begin{example}\label{ex-commafree-family}
Reconsider the obligation word $ftftt$. Let $Y=Y_{vert}=Y_{hor}$ be the comma-free binary code of length $k=5$: 
\[
Y =\{{\red 00111}, \green 00001, {\blue 10001}\}.
\] 
Let $X$ be a picture code as in Definition~\ref{def-family-comma-free}.  Figure~\ref{figCommaFreeFamily} shows a code-picture $p$ in $X^{2,2}$.

\begin{figure}[htbp]
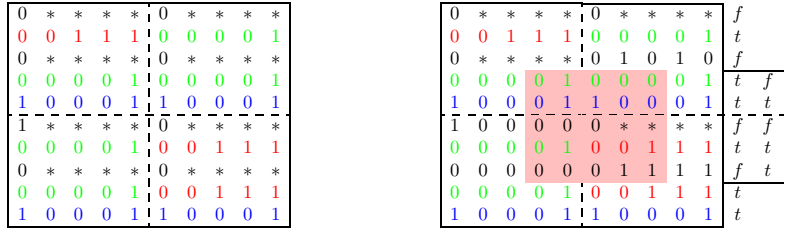

\setlength{\tabcolsep}{40pt}
\begin{center}\scalebox{0.7}{
	\begin{tabular}{ll}
		$ 
		\begin{array}{|lllll:}
			\hline
			0&*&*&*&* \\
			\red 0&\red 0&\red 1&\red 1&\red 1\\
			0&*&*&*&*\\
			\green 0&\green 0&\green 0&\green 0&\green 1\\
			\blue 1&\blue 0&\blue 0&\blue 0&\blue 1\\
						\hdashline
			1&*&*&*&*\\
			\green 0&\green 0&\green 0&\green 0&\green 1\\
			0&*&*&*&*\\
			\green 0&\green 0&\green 0&\green 0&\green 1\\
			\blue 1&\blue 0&\blue 0&\blue 0&\blue 1\\
			\hline
		\end{array}
		\begin{array}{lllll|}
 			\hline
			0&*&*&*&* \\
			\green 0&\green 0&\green 0&\green 0&\green 1\\
			0&*&*&*&*\\
			\green 0&\green 0&\green 0&\green 0&\green 1\\
			\blue 1&\blue 0&\blue 0&\blue 0&\blue 1\\
			\hdashline
			0&*&*&*&* \\
			\red 0&\red 0&\red 1&\red 1&\red 1\\
			0&*&*&*&*\\
			\red 0&\red 0&\red 1&\red 1&\red 1\\
			\blue 1&\blue 0&\blue 0&\blue 0&\blue 1\\\hline
		\end{array}
		$

		&
		%figura con la finestra interna e gli obligation words

		$ 
		\begin{array}{|lllll:}
			\hline
			0&*&*&*&* \\
			\red 0&\red 0&\red 1&\red 1&\red 1 \\
			0&*&*&*&*\\
						\cline{4-5}
			\green 0&\green 0&\green 0 &\cellcolor{pink}\green 0&\cellcolor{pink}\green 1\\
			\blue 1&\blue 0&\blue 0&\cellcolor{pink}\blue 0&\cellcolor{pink}\blue 1\\
						\hdashline
			1&0&0&\cellcolor{pink}0&\cellcolor{pink}0 \\
			\green 0&\green 0&\green 0&\cellcolor{pink}\green 0&\cellcolor{pink}\green 1\\
			0&0&0&\cellcolor{pink}0&\cellcolor{pink}0\\
									%\cline{4-5}
			\green 0&\green 0&\green 0&\green 0&\green 1\\
			\blue 1&\blue 0&\blue 0&\blue 0&\blue 1\\
			\hline
		\end{array}
		\begin{array}{lllll|c c}
						\cline{1-5}
			0&*&*&*&* & f \\
			\green 0&\green 0&\green 0&\green 0&\green 1 & t\\
			0&1&0&1&0 & f \\
									\cline{1-3} \cline{6-7} 
			\cellcolor{pink}\green 0& \cellcolor{pink}\green 0&\cellcolor{pink}\green 0&\green 0&\green 1& t & f\\
			\cellcolor{pink}\blue 1&\cellcolor{pink}\blue 0& \cellcolor{pink}\blue 0&\blue 0&\blue 1& t& t \\
			\hdashline %\cdashline{6-6}
			\cellcolor{pink}0&\cellcolor{pink}*&\cellcolor{pink}*&*&* & f & f \\
			\cellcolor{pink}\red 0&\cellcolor{pink}\red 0&\cellcolor{pink}\red 1&\red 1&\red 1& t& t \\
			\cellcolor{pink}0&\cellcolor{pink}1&\cellcolor{pink}1 &1&1 & f & t \\
					%\cline{1-3} 
					\cline{6-7} 
			\red 0&\red 0&\red 1&\red 1&\red 1& t\\
			\blue 1&\blue 0&\blue 0&\blue 0&\blue 1& t
			\\\cline{1-5}
		\end{array}
		$
	\end{tabular}
	}%scalebox
\end{center}
\caption{A picture $p$ in $X^{2,2}$. An asterisk stands for a bit not obliged to occur in a code-word. 
The right picture highlights a generic internal position where the presence of a code-picture (having the same obligation word) is impossible.}\label{figCommaFreeFamily}
\end{figure}
\par\noindent
Picture $p$ is obliged to have code-words at rows $ 2, 4, 5, 7, 9, 10$. 
It is not difficult to verify that any code-picture in $X$ may not occur as an internal subpicture of $p$, otherwise at least one of its rows (or column) holding a code-word would overlap a row (or a column) of $p$ containing two concatenated code-words.  
For instance, focus in the right picture on the $(5\times 5)$ subpicture highlighted, overlapping code-words at rows $ 4, 5, 7$ of $p$. 
Imagine to replace the subpicture with a code-picture in $X$, which by obligation has a code-word at row 5 (of $ p$): clearly such a code-word would cross two code-words of $p$ thus contradicting Definition~\ref{def-family-comma-free}.

\par\noindent
Let $Y_0\subset Y$ be the subset of code-words beginning with 0, and similarly for $Y_1$. 
Thus, any row $i$ is in $Y_0$ if the pixel $p_{i,1} = 0$, it is in $Y_1$ otherwise. 
It can be easily verified that the code numerosity is $|X| = \left( |Y_0|\cdot |Y_1|^2 +2|Y_0|^2\cdot |Y_1|\right)\cdot 2^8 = 2560$.
\medskip
\par
The second example is a comma-free code having the minimal value $k=3$ permitted by Definition~\ref{def-family-comma-free}.
	We define the code $X_3$ with $Y_{hor}= \{110, 100\}$, $Y_{vert}=\{011\}$ and obligation word $ftt$. 
	There are 16 code-pictures with those choices, represented as:
		\[ 
		\begin{array}{|ccc|} \hline
				0 & * & * 
				\\
				1 & 1&0 
				\\
				1 & 1 &0 
				\\ \hline
			\end{array}\,\, 
			\begin{array}{|ccc|} \hline
				0 & * & * 
				\\
				1 & 0&0 
				\\
				1 & 0 &0 
				\\ \hline
			\end{array}\,\,
			\begin{array}{|ccc|} \hline
				0 & * & * 
				\\
				1 & 1&0 
				\\
				1 & 0 & 0 
				\\ \hline
			\end{array}\,\,
			\begin{array}{|ccc|} \hline
				0 & * & * 
				\\
				1 & 0&0 
				\\
				1 & 1 & 0 
				\\ \hline
			\end{array}
	\]
where the pixels in subpicture $p_{\subd 1 2 {1}{3}}$ (marked with asterisk) may take any value. 
\end{example}
\medskip

\begin{proposition}\label{prop-commafree-family}
For every $k\geq 3$, the family of Definition~\ref{def-family-comma-free} exclusively includes comma-free picture codes. 
Moreover, if $k$ is prime, then there is a comma-free picture code in $\{0,1\}^{k,k}$ of numerosity at least: 
\[
\frac{2^{k^2-k}}{{(k+1)}^{2\sqrt{k}}}
\]
\end{proposition}
\proof 
\par\noindent{\em Part 1. }
The family only includes comma-free picture codes. 
Consider a picture $p= \mytileDynamic{x_{11}}{x_{12}}{x_{21}}{x_{22}}\in X^{2,2}$, with $x_{i,j}\in X$, and let $x \in X$. 
We show that $x$ cannot be an internal subpicture of $p$.
By contradiction, suppose that $x$ is an internal subpicture of $p$, with top left corner in $x_{11}$.
We consider just two cases for the coordinates $(i,j)$ of the top left corner of $x$ in $p$, since the remaining cases are symmetrical: 
\begin{enumerate}
	\item\label{en-case1} $1<i\le k$ and $j=1$;
	\item\label{en-case2} $1<i\le k$ and $1<j\le k$.
\end{enumerate} 
In case~\ref{en-case1}, the leftmost column of $x$ (a comma-free code of $Y_{vert}$) overlaps with the concatenation of the leftmost column of $x_{11}$ 
and the leftmost column of $x_{21}$ (both also being vertical codes), a contradiction with respect to the definition of comma-free code for words. 

\par\noindent In case~\ref{en-case2}, given the obligation word $w$, the comma-free 
code containing rows of $x_{11}$ and $x_{21}$ correspond to the values $t$ in $ww$. The rows of $x=p_{\subd i j {i+k}{j+k}}$ correspond in $p$ to the $(i-1)$-left-rotation of $w$. By definition of obligation word, 
there exists a row $h$ of $x$ that overlaps the concatenation either of two rows $h'$ of $x_{11}$ and $x_{12}$, 
or of two rows $h'$ of $x_{21}$ and $x_{22}$, which again contradicts the definition of comma-free code on words. 
\medskip
\par\noindent{\em Part 2. }
We prove a lower bound on the numerosity of a comma-free picture code of size $(k,k)$.
\par\noindent
Let $\Lambda=\{0,1\}$. We recall from~\cite{DBLP:journals/tit/Eastman65,DBLP:journals/dm/PerrinR18} that the numerosity of a binary comma-free word code $Y$ of length $k$, for $k$ prime, is 
\begin{equation}\label{eqNumerosityCFwordCode}
\nu= \frac{2^k-2}{k}
\end{equation}
a value to be later used.
We assume that at least $\frac{\nu}{2}$ of the codes in $Y_{hor}$ begin with $0$ (otherwise just exchange 0 and 1 in $Y_{hor}$). 
We define a word $w$ in terms of a set $Q\subset \{1, 2 , \ldots, k \}$ 
such that the $i$-th letter is $w_i= t$ if, and only if, $i \in Q$.
For a given $q<k/2$,
define the set $ Q\subset \{1, 2 , \ldots, k \}$ as follows:

\begin{equation}\label{eqShift}
\begin{array}{ll}
Q= 
 & \left\{i \mid 1\le i\le q \right\} \;\cup\, \\
 &\left\{i \mid q<i \le k \, \land \, i \Mod q =0 \right\}
\end{array}
\end{equation}
Hence, $w_i=t$ when $i=1, 2, \dots, q$ and when $i= 2q, 3q, \dots jq$, for $j$ such that $k-q<jq\le k$. 
It is easy to see by simple arithmetic considerations that the choice of $Q$ in~\eqref{eqShift} makes $w$ an obligation word. 
From the definition~\eqref{eqShift} of set $Q$, it follows that in any $p \in X$
the number of rows having a horizontal code is $q+k/q$. 
\par\noindent
It is immediate to notice that the number of such rows is minimal 
for $q=\lfloor \sqrt{k}\rfloor $, hence, their number is 
$\lfloor \sqrt{k}\rfloor+k/\lfloor \sqrt{k}\rfloor = 2\lfloor \sqrt{k}\rfloor =2q$. 
\par\noindent
To simplify the following computation, we choose a singleton vertical code $Y_{vert}=\{0^q10^{k-q-1}\}$, 
thus fixing the leftmost column of every code-picture in $X$. 
\par\noindent
Using the value $\nu$ from~\eqref{eqNumerosityCFwordCode}, the total number of possibilities using comma-free codes in $2q$ rows of a picture code is:  
\begin{equation}\label{eq-commafree-rows}
(\frac{\nu}{2})^{2q} = \left(\frac{2^{k}-2}{2k}\right)^{2q} = \left(\frac{2^{k-1}-1}{k}\right)^{2q}>\left(\frac{2^{k-1}}{k+1}\right)^{2q} = \frac{2^{2kq-2q}}{(k+1)^{2q}} .
\end{equation}
The number of rows free from horizontal codes in a code-picture is $(k-2q)$, 
each row containing $k-1$ free bits, for a total number of possibilities: 
\begin{equation}\label{eq-rowsWithoutCode}
\left(2^{k-1} \right)^{k-2q} =2^{k^2-k+2q-2qk}
\end{equation}
Multiplying~\eqref{eq-commafree-rows} by~\eqref{eq-rowsWithoutCode}, we obtain the following lower bound on the number of the picture codes in $X$: 
\begin{equation}
|X|\ge \frac{2^{2kq-2q}}{(k+1)^{2q}}\cdot 2^{k^2-k+2q-2qk}
 =  \frac{2^{k^2-k}}{{(k+1)}^{2q}}
\end{equation}
Substituting $\lfloor \sqrt{k}\rfloor$ for $q$ in the denominator we obtain: 
$(k+1)^{2q}= (k+1)^{2\lfloor \sqrt{k}\rfloor} \le (k+1)^{2\sqrt{k}}$, i.e.,

\begin{equation}\label{eq-pictureCodeLowerBound}
|X|\ge  \frac{2^{k^2-k}}{{(k+1)}^{2\sqrt{k}}}
\end{equation}
\qed
\medskip
\par
The lower bound of Proposition~\ref{prop-commafree-family} has been computed under simplifying but pessimistic assumptions, 
in particular that $|Y_{ver}|=1$, thus fixing the leftmost column of code-pictures. This significantly reduces the number of possible code-pictures for small values of $k$. 
For instance, with $k=5$ the value according to the lower bound in Eq.~\eqref{eq-pictureCodeLowerBound} is $2^3$, in contrast with 
the numerosity of Example~\ref{ex-commafree-family} where a comma-free vertical/horizontal code of just 3 elements yields 2560 code-pictures.

\par
To finish, it is obvious that the family of Definition~\ref{def-family-comma-free} does not exhaust all possible comma-free picture codes.
For instance, the (singleton) set: 
$$
X= 	\left\{\begin{array}{|cccc|} \hline
 1 & 0 & 0 & 0 
\\
 0 & 0& 0 & 0 
\\
 0 & 0 & 0 & 0 
\\
 0 & 0 & 0 & 0 
\\ \hline
 \end{array}\right\}
 $$  
is a comma-free picture code, although it does not comply with Definition~\ref{def-family-comma-free}.

 \begin{proposition}\label{prop-commafreeEnough} 
For every $m\ge 2$, there exist $k\geq 2$ and a comma-free code  $X\subset \{0,1\}^{k,k}$ such that $|X|\ge m^k$, with $k \in O(\lg m)$.

\end{proposition}
\proof
 Let $k$ be any prime number between $4\lg m$ and $8 \lg m$, which must exist by the Bertrand-Chebyshev theorem (see e.g.~\cite{HardyWright}, Chapter 22). Hence, $k$ is in $\mathcal{O}(\lg m)$. 
\par\noindent
From Proposition~\ref{prop-commafree-family}, there exists a comma-free code $X\subseteq \{0,1\}^{k,k}$ with cardinality
$$|X|\ge \frac{2^{k^2-k}}{{(k+1)}^{2\sqrt{k}}}.$$ 
Since 
$(k+1)^{2\sqrt{k}}= 2^{\log(k+1)^{2\sqrt{k}} }$ we obtain:
\begin{equation*}
|X| \,\geq\, 2^{k^2-k-2\sqrt{k} \log(k+1) }.
\end{equation*}
It can also be immediately derived that $|X| \,\geq\, 2^{k^2-2k}$,
since $k>2\sqrt{k}\log(k+1)$
when $k\ge 4\lg m \ge 4$. 
Hence, $k^2-2k$ bits are enough to represent all  code-pictures of $X$, 
 while we need $\lg(m^k)= k \lg m$ bits to define $m^k$.
Therefore, 
$$|X|\ge m^k\Rightarrow {k^2-2k}\ge {k \lg m} \iff
k\ge 2+\lg m.$$
\par\noindent
Since $k \ge 4\lg m$, we have that $|X|\ge m^k$, since $4\lg m\ge 2+\lg m$ for all $m\ge 2$.
\qed

\subsection{Strict local testability of encoded pictures}\label{ssect-SLTofEncodedPict}
The next proposition states in 2D the already mentioned SLT property of 1D comma-free 
codes~\cite{Restivo1974,HashiguchiHonda1976} further developed in~\cite{DeLucaRestivo1980}.
\begin{proposition}\label{prop-X-k-local}
Let $X \subseteq \Lambda^{k,k}$ be a comma-free picture code on words. The language $ X^{++}$ is $2k$-SLT.
\end{proposition}
\proof
We show that $L\left(\ktile{\widehat {X^{++}} }{2k}\right)=X^{++}$ whence the statement. 
\par\noindent
The left to right inclusion is obvious since if $p \in X^{++}$ then $\ktile{\widehat p}{2k}\subseteq \ktile{\widehat {X^{++}} }{2k}$.
\par\noindent
We prove the right to left inclusion.
Let $p \in L\left(\ktile{\widehat { X^{++}}}{2k}\right)$, and let $q\in\ktile{X^{++}}{2k}$ be a $2k$-tile of a (non-bordered) picture $z$ of size $(kr, kc)$.
We claim that if $q$ is such that the subpicture $q_{nw}=q_{\subd 1 1 {k} {k}}$ is in $X$, 
then $q \in X^{2,2}$ (i.e., it is tessellated by four code-pictures).
\par\noindent 
Since $X$ is comma free, if $q_{nw}\in X$, then $q_{nw}$ must coincide with one of the subpictures in the $k$-tessellation of $z$,
 otherwise $q_{nw}$ would be an internal subpicture of a $2k$-tile in $X^{2,2}$, against Definition~\ref{def-commafree}. Therefore, $q \in X^{2,2}$.
\par\noindent
Consider now the $(k+1)$-tile of $\widehat p$ positioned at the north-west corner, which has the form:
\[
\mytileDynamic{\#}{\#^{\obar k}}{\#^{\ominus k}}{x_{11}} \quad \text{ for some } x_{11} \in X.
\]
By the above claim, a $2k$-tile having $x_{11}$ in the north-west position must be in $X^{2,2}$.
We enlarge the north-west code-picture $x_{11}$ of $p$ towards east and south, into a $(2k+1, 2k+1)$ picture, that must have the following form:
\begin{equation}\label{eq-tedious_induction}
\myLargeTile{\#}{\#^{\obar k}}{\#^{\obar k}}{\#^{\ominus k}}{x_{11}}{x_{12}}{\#^{\ominus k}}{x_{21}}{x_{22}} \qquad \text{ with } x_{i,j}\in X 
\end{equation}
i.e., also $\mytileDynamic{x_{11}}{x_{12}}{x_{21}}{x_{22}}\in X^{2,2}$.
A simple but tedious induction would permit to enlarge the picture in Eq. ~\eqref{eq-tedious_induction}, thus proving that $p$ is in $X^{++}$.
\qed
\medskip

\par
At last we consider a local language (defined by a set of 2-tiles) and we encode each symbol using a comma-free picture code. 
The following theorem states that the resulting language is SLT (an analogous property for words is stated in~\cite{HashiguchiHonda1976}).
\begin{theorem}\label{th-SLTCode}
Let $\tset\subseteq \Gamma^{2,2}$ be a set of 2-tiles defining the local language $L(\tset )$ and let $X\subseteq\Lambda^{k,k}$ be a comma-free picture code such that $|X|=|\Gamma|$. 
The encoding $\codeX{L(\tset)}$ is a $2k$-SLT language. 
\end{theorem}
\proof Let $\overline{\tset}$ be the complement of $\tset$, i.e, 
$\overline{\tset} = \Gamma^{2,2}-\tset$, which can be interpreted as the set of ``forbidden" $2$-tiles of $L(\tset)$.
Let $\tsetM_{2k}= \ktile{\widehat {X^{++}} }{2k}-\codeX{\overline{\tset}}$.
To prove the thesis we claim:
\[
L(\tsetM_{2k}) = \codeX{L(\tset)}\,.
\]
\par\noindent
First, we prove the right to left inclusion.
Let $p \in \codeX{L(\tset)}\subseteq X^{++}$, hence there exists $q \in \Gamma^{++}$ such that $p = \codeX{q}$. 
If picture $q$ has size $(r,c)$, with $r,c\ge 1$, then 
$p$ has size $(kr,kc)$; each element of the $k$-tessellation of $p$ can be denoted as $x_{i,j}=\codeX{q_{i,j}}$. 
\par\noindent
By contradiction, assume that $p\not\in L(\tsetM_{2k})$; hence, there is a $2k$-tile $\rho\in\ktile{\widehat p}{2k}$ such that $\rho \not\in \tsetM_{2k}$.
Since $p \in X^{++}$, it must be $\rho\in \ktile{\widehat{ X^{++}}}{2k}$. Moreover, by definition of $\tsetM_{2k}$, 
$\rho \in \codeX{\overline{\tset}}$. Therefore, $\rho= \codeX{\overline{\theta}}$ for some $\overline{\theta}\in \overline{\tset}$, thus $\rho \in X^{2,2}$. 
\par\noindent
Since $X$ is a comma-free picture code, no subpicture in $X$ of $\rho$ can be an internal subpicture of the $2k$-tiles in $X^{2,2}$ of $p$, hence
$\rho=\mytileDynamic{x_{i,j}}{x_{i,j+1}}{x_{i+1,j}}{x_{i+1,j+1}}$ for some $i,j$. It follows that $\rho= \codeX{\theta}$ for 
$\theta = \mytileDynamic{p_{i,j}}{p_{i,j+1}}{p_{i+1,j}}{p_{i+1,j+1}}\in \tset$. 
Since $\codeX{}$ is one-to-one, $\rho$ cannot also be equal to $\codeX{\overline{\theta}}$ for $\overline{\theta}\neq\theta$, a contradiction.
\smallskip
\par\noindent
Next we prove the left to right inclusion.
Let $p\in L(\tsetM_{2k})$. 
Since, by definition of $\tsetM_{2k}$, $p\in X^{++}$ and, by Proposition~\ref{prop-X-k-local}, $X^{++}$ is $2k$-SLT, 
we have that $p$ has size $(kr,kc)$, with the $k$-tessellation of $p$ defined by the subpictures $ x_{i,j}\in X$. 
\par\noindent
Since $\codeX{}$ is a bijection from $\Gamma$ to $X$, there exists one, and only one, symbol in $\Gamma$, denoted as $\gamma_{i,j} \in \Gamma$, such that 
$x_{i,j}=\codeX{\gamma_{i,j}}$. 
Therefore, we can define a picture $q$ such that $q_{i,j}= \gamma_{i,j}$, with $p = \codeX{q},x_{i,j} = \codeX{\gamma_{i,j}}$. 
\par\noindent
Consider a $2k$-tile $\xi \in \tsetM_{2k}\cap X^{2,2}$, denoted by $\xi=\mytileDynamic{x_{i,j}}{x_{i,j+1}}{x_{i+1,j}}{x_{i+1,j+1}}$.
\par\noindent 
Since $\xi\not\in\codeX{\overline{\tset}}$, it must be $\xi\in\codeX{\tset}$, i.e, there is $\theta \in \tset$ such that
$ \xi= \codeX{\theta}$, with $\theta = \mytileDynamic{\gamma_{i,j}}{\gamma_{i,j+1}}{\gamma_{i+1,j}}{\gamma_{i+1,j+1}}\in \tset$. 
Therefore, all the tiles in $q$ are in $\theta$, hence $q \in L(\tset)$.
\qed

\section{Main result}\label{s-main} 

Before we present the main result that any recognizable language is the projection of an SLT language having alphabetic ratio two, 
 we show that for some language in \rec~a ratio smaller than two does not suffice. This negative statement reproduces in 2D the statement and the proof for regular word languages in ~\cite{DBLP:journals/ijfcs/Crespi-ReghizziP12}, Theorem 5.
\subsection{The minimal alphabetic ratio}
\begin{theorem}[minimal alphabetic ratio]
There exists a TS recognizable languag  $R$ over an alphabet $\Sigma$ such that for every $k$-tiling system $(\Sigma, \Gamma, \tset_k, \pi)$ such that $R=\pi(L(\tset _k))$, the alphabetic ratio is $\frac{|\Gamma|}{|\Sigma|}\ge 2$.
 %$\Gamma$ 
\end{theorem}
\proof
For a generic letter $a$, let $R_a$ be the language of all square pictures over $\{a\}$, of size at least $(2,2)$. 
It is obvious that $R_a$ can only be recognized by tiling systems having a local alphabet $\Gamma$ of cardinality at least $2$. In fact, if $|\Gamma|=1$, then 
a non-square (rectangular) picture and a square picture can be covered by the same set of tiles. 
\par\noindent
Let $\Sigma = \{ b,c\}$; we prove the thesis for $R=R_b \cup R_c$. If $|\Gamma|< 4$, then consider two pictures: $p' \in R_b, p'' \in R_c$. 
Let $ \beta, \gamma \in \Gamma^{++}$ be their respective pre-images. 
Since $p'$ only includes symbol $b$, every symbol of $\beta$ must be projected to $b$; similarly, every symbol of $\gamma$ must be projected to $c$. 
Since $|\Gamma|< 4$ (e.g. $|\Gamma|=3$) and the symbols in $\beta$ must be different from the symbols in $\gamma$, one of the two pictures, say, $\beta$, 
must be composed of just one type of symbol (i.e., it is on a unary alphabet), but we already noticed that each $R_a$ requires two local symbols.
 %it is impossible. %; therefore, $\beta$ has the same $2$-tiles of a non-square rectangular picture, a contradiction with the definition of $R_b$.
The generalization to an alphabet $\Sigma$ of larger cardinality is immediate, by considering $R= \bigcup_{a\in\Sigma} R_a$.
\qed
\medskip
\par
The above theorem leaves open the possibility that the alphabetic ratio two may suffice for all recognizable languages. This is proved in Section~\ref{ssec-MET}.
\par

\subsection{Padded picture languages}
In later proofs it is convenient to adjust the picture height and width to be a multiple of the same integer $k\geq 2$, in order to apply a $k$ tessellation. 
To this end, we introduce a transformation, called padding, that sets a given picture into the north-west corner of a sufficiently larger picture having both sides mutiple of $k$. 
The transformation respectively appends to the east and to the south side of the picture some columns and rows, filled with a new letter not present in the original alphabet. 
\par
More precisely, let $R \subseteq \Sigma^{++}$ be in \rec~, and let $k\geq 2$. Intuitively, we define a language 
$R^{(k)}\subseteq (\Sigma \cup \{\$ \})^{++}$, where $\$ \notin \Sigma$, obtained 
by concatenating vertically and then horizontally 
each picture of $R$ with two rectangular pictures in $\{\$\}^{++}$, of minimal size, 
such that the resulting picture has size $(m, n)$, where both $m$ and $n$ are multiple of $k$.
The reader may look at Figures~\ref{fig-tesselation},~\ref{fig-padded} for two padded pictures where $k=3$. 
The formal definition follows.

\begin{definition}[Padded language]\label{def-padded}
Let $R \subseteq \Sigma^{++}$ be in \rec~and let $k\geq 2$. 
	Let $ V_k, H_k \subseteq \{\$ \}^{++}$ be the languages such that:
	\[
	 V_k =\left\{ \{\$\}^{n, h} \mid n>0,1\le h\le k\right\}\; \text{ and }\; H_k =\left\{ \{\$\}^{h, n} \mid n>0,1 \le h\le k\right\}.
	\]
	Then the \emph{padded language}, denoted by ${R}^{(k)}$, over the alphabet $\Sigma_\$ = \Sigma \cup
	\{\$\}$ is:
	\begin{equation}\label{eq:paddedLanguage2}
		{R}^{(k)\,} =\,
		\begin{array}{|c|c|}
			\hline
			R & \multirow{2}{*}{$ H_k$} \\ \cline{1-1}
			 V_k &
			\\ \hline
		\end{array} \,\cap \left( \left( \Sigma_\$ \right)^{k,k} \right)^{++}.
	\end{equation}
\end{definition}
\medskip
	Notice that the definition of  padded language is such that every picture has always at least one padded row and one padded column, 
	and at most $k$ padded rows and $k$ padded columns. It is easy to see that both $\left|{R}^{(k)\,} \right|_{row}$ and $\left|{R}^{(k)\,} \right|_{col}$ are multiple of $k$.

\par
	Looking again at Figure~\ref{fig-padded}, the original picture has size $(6,12)$, therefore the horizontal and the vertical padded borders have thickness $3$; in Figure~\ref{fig-tesselation}, since the original pictures has size $(7,14)$,
the horizontal border has thickness $2$ and the vertical border has thickness $1$, so that the padded picture has size $(9,15)$.

\begin{figure}
\setlength{\tabcolsep}{10pt}
\renewcommand{\arraystretch}{1.0}
\medskip
\begin{center}
\scalebox{0.9}{
$
\begin{array}{ |ccc: ccc: ccc: ccc: ccc|} 
\hline
 a &a &a &a &a &a &a &a &a &a &a &a &\$ & \$ &\$\\
 a &a & a &a &a &a&a &a&a & a & a &a &\$ & \$ &\$
\\
 a &a &a & a & a& a &a&a &a &a &a &a &\$ & \$ &\$
\\ \hdashline
 a &a &a &a&a& a &a &a &a &a& a &a & \$ & \$ &\$
\\
 a &a &a &a &a&a&a &a &a &a &a & a &\$ & \$ &\$
\\
 a &a &a &a &a &a&a &a &a &a &a & a &\$ & \$ &\$
\\ \hdashline
 \$ & \$& \$ & \$& \$ &\$& \$&\$& \$ & \$ & \$ &\$ & \$ & \$ &\$
\\ \$ & \$& \$ & \$& \$ &\$& \$&\$& \$ & \$ & \$ &\$ & \$ & \$ &\$
\\ \$ & \$ & \$ & \$ & \$ &\$& \$&\$& \$ & \$ & \$ &\$& \$ & \$ &\$
\\
\hline
\end{array} 
$ 
}
\end{center}
\caption{The 3-padded picture of $a^{6,12}$ in language $R$ of Example~\ref{ex:k-localLanguages}.}\label{fig-padded} 
\end{figure}

Since \rec~is closed with respect to concatenations and intersection, given $R$ in \rec, it follows from Eq.~\eqref{eq:paddedLanguage2} that also $R^{(k)}$ is in \rec , and
 there is a standard construction of a TS for $R^{(k)}$ that we do not use because it is not optimal, in the following sense. 
Such a TS uses a large local alphabet, with cardinality of the order of $|\Gamma|\cdot k^2$, where $\Gamma$ is the local alphabet of a TS for $R$. 
\par
On the other hand, the following theorem shows, by means of a more focused construction, that a smaller local alphabet suffices. 
This will be important in the proof of the main result. 
\begin{theorem}\label{thm-size-kSLT-padded}
If language $R\subseteq \Sigma^{++}$ is defined by a TS with local alphabet $\Gamma$, for all $k\geq 2$ the padded language $R^{(k)}\subseteq (\Sigma \cup \{\$ \})^{++}$ can be defined by a TS with a local alphabet of size 
	$|\Gamma| + k+1$. 
\end{theorem}
\proof
 Let $ (\Sigma,\Gamma',\tset',\pi')$ be a TS recognizing $R$. We construct the TS 
$(\Sigma_{\$}, \Gamma, \tset,  \pi)$ that recognizes $R^{(k)}$.
\par\noindent
Let $\Delta=\{\flat, 1, \dots, k\}$ be a new alphabet, disjoint from $\Gamma$. 
Define the pictures $p^{(i,k)}\in \Delta^{++}$, of size $(i,k)$ with $1\le i\le k$, such that every pixel is $\flat$, except for the rightmost column and the bottom row;  the rightmost column is 
 $k-i+1 \ominus k-i \ominus \dots\ominus k$, 
 and the bottom row is $1 2 \dots k$. 
\par\noindent For instance $p^{(4,5)}$ is the picture: 
\begin{center}
\scalebox{0.8}{ 
$
\begin{array}{|ccccc|}
			\hline
			\flat&\flat&\flat&\flat&2 \\
 \flat&\flat&\flat&\flat&3 \\
			\flat&\flat&\flat&\flat&4 \\
			1&2&3&4&5\\
\hline
		\end{array}
		$
}
\end{center}
and $p^{(2,5)}$ is the picture: 
\begin{center}
\scalebox{0.8}{ 
$
\begin{array}{|ccccc|}
			\hline
			\flat&\flat&\flat&\flat&4 \\
			1&2&3&4&5\\
\hline
		\end{array}
		$
}.	
\end{center}

\par\noindent		
We define the  sets:
\[
\begin{array}{ll}
 \Delta_H &= \left\{ \Delta^{\ominus i}\mid 1\le i \le k\right\}^{\obar +} 
 \\
 \Delta_V &= \left\{ \Delta^{\obar i}\mid 1\le i \le k\right\}^{\ominus +} 
\\
H &= \bigcup_{1\le i \le k} \{p^{(i,k)}\}
\\
V&=\bigcup_{1\le i \le k} \{p^{(k,i)}\}.
\end{array}
\]
The set $\Delta_H$ is composed of pictures over the padding alphabet $\Delta$ having a number of rows in the interval $1, \dots, k$ and any number of columns. Symmetrically for $\Delta_V$. The set $H$ is a set of rectangular pictures of the form $p^{(i,k)}$, i.e., having a number $i$ of rows in the interval $1, \dots, k$ and $k$ columns.
\par\noindent
Let $\Gamma= \Gamma' \cup \Delta$ and let $L\subset \Gamma^{++}$ be the language defined by the formula:
\[
L= \, \begin{array}{|c|c|}
			\hline
			L(\tset ) & \multirow{2}{*}{$V^{\ominus+}$} \\ \cline{1-1}
			\Delta_H &
			\\ \hline
		\end{array} \;\cap \;
		\begin{array}{|c|c|}
			\hline
			L(\tset ) & \Delta_V\\ \hline
			\multicolumn{2}{|c|}{H^{\obar +}} 
			\\ \hline 
		\end{array}.
\]
Language $L$ is such that its bottom row and its leftmost column are words in $(1\, 2 \dots k)^+$, hence every picture of $L$ has both height and length multiple of $k$.
\par\noindent
By extending the projection $\pi'$ from $\Gamma'\to \Sigma $ to $\Gamma \to  \Sigma_{\$}$, 
with $\pi(a) =\pi'(a)$ for every $a \in \Sigma$ and  $\pi(\flat)=\pi(1) = \dots = \pi(k)=\$$, we obviously obtain that $\pi(L)=R^{(k)}$ and 
it remains to prove that $L$ is a local language, but we pause to show a helpful example.
\par 
Figure~\ref{fig-local-padding} illustrates the language $L$, applied as pre-image of the padded language $R^{(3)}$ associated with the language $R$ of Example~\ref{ex:k-localLanguages}.
It shows two pictures $p_0$ and $p_1$ over the alphabet $ \Gamma= \Gamma'\cup \Delta = \{n,\ssearrow, \nnearrow\} \cup \{1, 2, 3, \flat\}$ 
which is then mapped onto $\{a \cup \$\}$ to obtain the padded language. Picture $p_0$ has a number of "unpadded" rows multiple of 3, 
while in $p_1$ the number of unpadded rows is equal to $1 \Mod 3$. 
The projection 
 $\pi:\Gamma \to \{a , \$\}$ is such that $\pi(c)=a$ for $c\in \{n,\ssearrow, \nnearrow\}$ and $\pi(c)=\$$ for $c\in \{1, 2, 3, \flat\}$. 
 The images under $\pi$ of $p_0$ and $p_1$ are the pictures shown in Figure~\ref{fig-padded}. 

 \begin{figure}

\setlength{\tabcolsep}{9pt}
\begin{center}
\renewcommand{\arraystretch}{1.0}
\scalebox{0.9}{\begin{tabular}{c}
$
p_0 = \,
\begin{array}{| ccc: ccc: ccc: ccc: ccc|} \hline
 \ssearrow & n & n & n & n & n & n & n & n &n&n& \nnearrow &\flat & \flat &1\\
 n & \ssearrow &n & n & n & n& n & n& n &n & \nnearrow & n &\flat & \flat &2
\\
 n & n & \ssearrow &n & n&n & n&n&n& \nnearrow & n & n &1 & 2 &3
\\ \hdashline
 n & n & n & \ssearrow& n&n &n&n& \nnearrow & n&n & n & \flat & \flat &1
\\
 n & n & n & n &\ssearrow& n&n& \nnearrow & n & n & n &n &\flat & \flat &2
\\
 n & n & n & n & n &\ssearrow&\nnearrow &n& n & n & n &n &1 & 2 &3
\\ \hdashline
 \flat & \flat& 1 & \flat& \flat &1& \flat&\flat& 1 & \flat & \flat &1 & \flat & \flat &1
\\ \flat & \flat& 2 & \flat& \flat &2& \flat&\flat& 2 & \flat & \flat &2 & \flat & \flat &2
\\ 1 & 2 & 3 & 1 & 2 &3& 1&2& 3 & 1 & 2 &3& 1 & 2 &3
\\
\hline
\end{array} 
$ 
\\
\\
$
p_1 = \,
\begin{array}{| ccc: ccc: ccc: ccc: ccc|} \hline
 \ssearrow & n & n & n & n & n & n & n & n & n & n &n&n& \nnearrow &1\\
 n & \ssearrow &n & n & n & n & n& n & n& n &n &n& \nnearrow & n &2
\\
 n & n & \ssearrow &n & n & n & n&n & n&n&n& \nnearrow & n & n &3
\\ \hdashline
 n & n & n & \ssearrow&n & n & n&n &n&n& \nnearrow & n&n & n & 1
\\
 n & n & n & n &\ssearrow&n &n& n&n& \nnearrow & n & n & n &n &2
\\
 n & n & n & n & n &\ssearrow&n& n&\nnearrow &n& n & n & n &n &3
\\ \hdashline
 n & n & n & n &n& n &\ssearrow& \nnearrow &n& n & n & n &n &n &1
\\ \flat & \flat& 2 & \flat& \flat &2& \flat&\flat& 2 & \flat & \flat &2 & \flat & \flat &2
\\ 1 & 2 & 3 & 1 & 2 &3& 1&2& 3 & 1 & 2 &3& 1 & 2 &3
\\
\hline
\end{array}
$
\end{tabular}
}
\end{center}
 \caption{Two (pre-image) pictures, $p_0$ and $p_1$, in the local language $L$ of the TS defining the padded language $R^{(3)}$ 
 associated with the language $R$ of Example~\ref{ex:k-localLanguages}. The TS is defined in the proof of Theorem~\ref{thm-size-kSLT-padded}.}\label{fig-local-padding}
\end{figure}

To prove the thesis, we construct the tiles that define  the language  $L$.
The following set $\tset $ of tiles is  composed of   four disjoint parts:
For conciseness we call \emph{border tile} a tile that includes a $\sharp$ pixel.
\begin{description}
	\item[tiles over $\Gamma'$: ] all tiles of $\tset ' \cap (\Gamma')^{2,2}$ are in $\tset $.
	
	\item[border tiles over $\Gamma' \cup \{\sharp\}$: ] all border tiles of $\tset '$ 
	having neither the east border nor the south border are also in $\tset $. 
	E.g., a tile $\mytile{\#}{\#}{\gamma_1}{\gamma_2}$ is in $\tset $ but not a tile $\mytile{\#}{\gamma_1}{\#}{\#}$, where $\gamma_1, \gamma_2 \in \Gamma'$.
	
	\item[border tiles over $\{\# \}\cup \Delta$ :] 
	the following tiles are in $\tset $:
	\[
\begin{array}{l}
	\left\{\mytile{\#}{1}{\#}{\#}, \mytile{\#}{\#}{1}{\#},\mytile{k}{\#}{\#}{\#},\mytile{\#}{\flat}{\#}{\flat}, \mytile{\flat}{\#}{\flat}{\#} \mytile{\#}{\flat}{\#}{1}, \mytile{\flat}{\flat}{\flat}{\flat}\right\}
	\\ \cup 
	\\ 
	\left\{ \mytile{i}{\#}{j}{\#}, \mytile{\flat}{i}{\flat}{j},\mytile{i}{\flat}{j}{\flat}, \mytile{i}{j}{\#}{\#},\mytile{\flat}{\flat}{i}{j} \mid 1\le i<k, j=i+1\right\}
\end{array}
\]
	
	\item[border tiles over $\Gamma' \cup\{\sharp\}\cup \Delta$ :]	
	for  all $\gamma, \gamma_1, \gamma_2\in \Gamma$, 
\begin{raggedright}
$\mytile{\#}{\gamma}{\#}{\#} \in \tset'\Rightarrow \left\{\mytile{\#}{\gamma}{\#}{\flat}, \mytile{\#}{\gamma}{\#}{1}\right\} \subseteq \tset$. 

$\mytile{\#}{\#}{\gamma}{\#} \in \tset'\Rightarrow \left\{\mytile{\#}{\#}{\gamma}{\flat}, \mytile{\#}{\#}{\gamma}{1}\right\}\subseteq \tset$.

$\mytile{\gamma}{\#}{\#}{\#} \in \tset'\Rightarrow \left\{\mytile{\gamma}{\#}{\flat}{\#}, \mytile{\gamma}{\#}{k}{\#}\right\}\subseteq \tset$.

$\mytile{\gamma_1}{\gamma_2}{\#}{\#} \in \tset' \Rightarrow% \text{ then } 
\left\{\mytile{\gamma_1}{\gamma_2}{\flat}{\flat}, \mytile{\gamma_1}{\gamma_2}{i}{j}\mid 1\le i\le k, j=1+\, i  \Mod k \right\}\subseteq \tset$.

\nopagebreak
$\mytile{\gamma_1}{\#}{\gamma_2}{\#} \in \tset'\Rightarrow
\left\{\mytile{\gamma_1}{\flat}{\gamma_2}{\flat} , \mytile{\gamma_1}{i}{\gamma_2}{j} \mid 1\le i\le k, j=1+\, i  \Mod k \right\}\subseteq \tset$.
\end{raggedright}
\end{description}
For brevity we omit the trivial proof that $L=L(\tset )$.
\qed

\begin{figure}
$$ \begin{array}{cccc}
\begin{array}{| ccc|} \hline
	\ssearrow & n & n\\
	n & \ssearrow &n \\
	n & n & \ssearrow \\ \hline
\end{array} &
\begin{array}{| ccc|} \hline
	n & n & n\\
	n & n &n \\
	n & n & n \\ \hline
	\end{array} &
		\begin{array}{| ccc|} \hline
			n & n & \nnearrow\\
			n & \nnearrow &n \\
			\nnearrow & n & n\\ \hline
		\end{array} &
		\begin{array}{| ccc|} \hline
			\flat&\flat & 1\\
			\flat & \flat&2 \\
			1 & 2 & 3\\ \hline
		\end{array}
\end{array}	
$$
\caption{The 3-tiles in $\kblock{p_0}{3}$ for the picture $p_0$ of Figure~\ref{fig-local-padding}.}\label{fig-3tiles}
\end{figure}

\subsection{ Extended Medvedev's Theorem for pictures}\label{ssec-MET}
By the padding technique, we  transform a tiling recognizable
 language $R \subseteq \Sigma^{++} $ into a language 
 $R^{(k)} \subseteq (\Sigma_{\$})^{++}$ 
 whose pictures have both sides multiples of $k$. 
$R^{(k)}$ is recognized by a tiling system $(\Sigma_{\$} , \Gamma, \tset, \pi)$,
 i.e., $R^{(k)} = \pi (L(T))$. 
 The pictures of $L(T)$ can be tessellated by subpictures belonging to the set $\kblock{L(T)}{k}$. 
This means that the elements of $L(T)$ may be also seen as pictures over 
the alphabet $\kblock{L(T)}{k}$. 
\par
A central observation is that, when one 
assembles the elements of $\kblock{L(T)}{k}$ in order to obtain the pictures 
of $L(T)$, the correctness of the assembly  only depends on the pixels located in the ``periphery'' of 
such elements. This motivates the introduction of the concept of frame, to formalize the idea of periphery. We associate to each picture $r \in \kblock{L(T)}{k}$ 
a pair $(f(r), \pi(r))$, where $f(r)$ is the frame of $r$ 
(see definition below) and $\pi(r)$ is the projection of r 
on the terminal alphabet $\Sigma$. 
If we denote by $B_k$ the set of 
such pairs for all $r$ in $\kblock{L(T)}{k}$, the language $R^{(k)}$ 
can be expressed as the projection of a local language over 
the alphabet $B_k$. This is formalized by the following Lemma~\ref{lemma-main}.
 
\par
For any picture $p \in \Gamma^{k,k}$, define the \emph{frame}, denoted by $f(p)$, as the quadruple of words: 
\begin{equation}
\label{eq:frame}
f(p) = (n_p, e_p, s_p, w_p), \quad n_p, e_p, s_p, w_p \in \Gamma^k 
\end{equation}
such that $n_p$ is the subpicture $p_{(1,1 ; k, 1)}$ (north row), $e_p$ is $p_{(k,1 : k,k)}$ (east column), and similarly $s_p, w_p$ are respectively the south row and west column. 
(The four words are not independent since each corner of $p$ is shared by two of the words.)
\par\noindent An example is in Figure~\ref{fig:frame}.

\begin{figure}
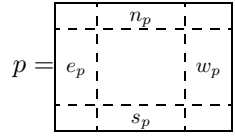

	\begin{center} $p =$\scalebox{0.8}{$ %\scalebox{1.5}{p}
			\begin{array}{|c:ccc:c|}				
	\hline &&n_p&&\\\cdashline{1-5}
	&&&& \\
		 e_p&& && w_p \\
		 &&&& \\
				 \cdashline{1-5}
				&&s_p & & \\ \hline
			\end{array}
			$
		}
	\end{center}
	\caption{The frame $f(p)$ of a picture $p\in \Gamma^{k,k}$. It is composed by the four words $(n_p, e_p, s_p, w_p)$ (which each corner being shared by two words).}\label{fig:frame}
\end{figure}

 \begin{lemma}\label{lemma-main}
		Let $k\ge 2$ and, for any $R\subseteq \Sigma^{++}$ in \rec, let $R^{(k)}$ be its padded language 
		recognized by the tiling system $(\Sigma_\$, \Gamma, \tset,  \pi)$. There exist a finite alphabet 
		$B_k\subseteq \Gamma^{4k}\times \Sigma_\$^{\,k,k}$, a set $M$ of 2-tiles over $B_k \cup \{\#\}$, and a morphism $\pi_k:B_k^{++}\to\Sigma_\$^{\,++}$ such that $R^{(k)}=\pi_k(L(M))$.
	\end{lemma}

	\proof 
 Consider the set ${\kblock{L(\tset )}{k}}$ obtained by the $k$-tessellation of the elements of $L(\tset )$. 
We define a new alphabet $B_k$ that for any element $r \in {\kblock{L(\tset )}{k}} $ includes the pair $\langle f(r),\pi(r)\rangle$ as a symbol, i.e., 
$B_k = \left\{\langle f(r),\pi(r)\rangle \mid r \in \kblock{L(\tset )}{k}\right\}$. 
\par\noindent
It is convenient to denote by $Q_k$ the set of the frames of the pictures in ${\kblock{L(\tset )}{k}}$: 
\begin{equation}\label{eq-Qk}
Q_k = \left\{f(r) \mid r \in \kblock{L(\tset )}{k}\right\} \subseteq \Gamma^{4k}; \quad \text{ therefore } |Q_k|\le |\Gamma|^{4k}.
\end{equation}
		Let $\varphi$ be the mapping associating each element $x=\langle f(r),\pi(r)\rangle \in B_k$ with its first component, i.e., 
	$\varphi(x)=f(r)$ (which is a frame). In the sequel, when no confusion can arise, $\varphi(x)$ is simply denoted by 
	$(n_x,e_x,s_x,w_x)$. 
	\par\noindent
	Let $\pi_k:B_k^{++}\to\Sigma_\$^{++}$ be the morphism defined by associating 
	each element $x=\langle f(r),\pi(r)\rangle \in B_k$ with its second component: 
	$\pi_k(x)=\pi(r)$.

\begin{figure}
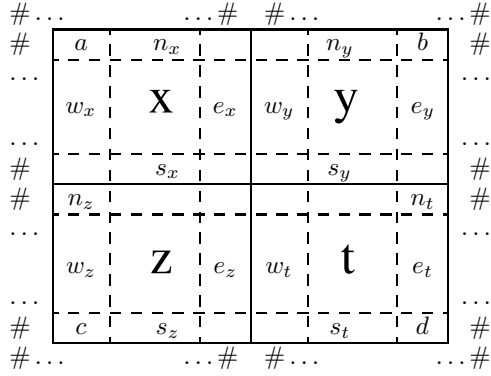

\begin{center}\scalebox{0.95}{$
		\begin{array}{l|c:cc:c|c:cc:c|r}
		\multicolumn{2}{l}{\#\dots }&& \multicolumn{2}{r}{\dots \#}&\multicolumn{3}{l}{\#\dots }& \multicolumn{2}{r}{\dots \#}\\ 
		\cline{2-9}
		\#&a&&n_x& & &n_y&&b&\# \\ \cdashline{2-9}
		\dots&&&&& & &&&\dots \\ 
		 & w_x&& \scalebox{2}{x\ } & e_x &w_y& \scalebox{2}{\  y} &&e_y& \\
\dots&&&&& & &&&\dots 
		\\ \cdashline{2-9}
		\#&& &s_x & &&s_y& &&\# 
		\\ \cline{2-9}
		\#&n_z&& & &&&&n_t&\# \\ \cdashline{2-9}
		\dots&&&&& & &&&\dots \\ 
		&w_z&& \scalebox{2}{z\ } & e_z &w_t&\, \scalebox{2}{\ t} &&e_t& \\
 \dots&&&&& & &&&\dots 
		\\ \cdashline{2-9}
		\#&c& &s_z& & &s_t&& d&\# \\\cline{2-9}
		\multicolumn{2}{l}{\#\dots }&& \multicolumn{2}{r}{\dots \#}&\multicolumn{3}{l}{\#\dots }& \multicolumn{2}{r}{\dots \#}\\
		\end{array}
		$
	}
\end{center}
\caption{A bordered picture of size $(2k,2k)$. It is composed by four $k$-tiles $x,y,z,t$ whose frames (e.g., $n_x,e_x, s_x, w_x$ for subpicture $x$) are evidenced by dashed lines. 
The symbols $a,b,c,d$ in the four corners of the picture are also evidenced.}\label{fig:four-k-tiles}
\end{figure}

\par\noindent
The definition of the set $M$ of 2-tiles, over the alphabet 
 $B_k \cup \{\#\}$, 
that we now introduce, translates the constraints on the adjacency in $L(\tset )$ of the elements of 
${\kblock{L(\tset )}{k}}$ in terms of the elements of $B_k$. 
The idea is that such an  adjacency is determined only by the tiles 
of $\tset $ that overlap the \emph{frames} (i.e. the periphery) of two adjacent $k$-tiles. 
\par\noindent
Now, we define the set $\tsetM_2\subseteq (B_k \cup \{\#\})^{2,2}$ of 2-tiles over the alphabet $B_k \cup \{\#\}$, distinguishing between 
{\em internal}, {\em border} and {\em corner} tiles.
\par\noindent
The bordered picture of size $(2k,2k)$, composed of four $k$-tiles $x, y, z, t$  
shown in Fig.~\ref{fig:four-k-tiles} may clarify the notation used below. 

\par\noindent
Let $x,y,z,t$ be in $B_k$.
\paragraph{Internal tiles}
 The 2-tile $\mytile{x}{y}{z}{t}\in \tsetM_2$ if the subpictures of size $(2,2)$ identified by the neighboring frames of the $k$-tiles $x,y,z,t$ are in $\tset $.
This is formalized by requiring that
\[ 
\begin{array}{l}
\tile{(s_x \obar s_y)\ominus (n_z\obar n_t)} \subseteq T \text{ and } 
\\ 
 \tile{(e_x \ominus e_z)\obar (w_y\ominus w_t)} \subseteq T.
\end{array}
\]
\paragraph{Border tiles}
\begin{description}[itemsep=2pt]
 \item $\mytile{\#}{\#}{x}{y}\in \tsetM_2$ iff $ \tile{\#^{2k\obar}\ominus(n_x \obar n_y)}\in \tset$; 
 \item $\mytile{\#}{x}{\#}{z}\in \tsetM_2$ iff ${\tile{\#^{2k\ominus}\obar (w_x \ominus w_z)}\in \tset}$;
 \item $\mytile{z}{t}{\#}{\#}\in \tsetM_2$ iff $\tile{(s_z \obar s_t)\ominus \#^{2k\obar}}\in \tset$;
 \item $\mytile{x}{\#}{z}{\#}\in \tsetM_2$ iff $\tile{(e_x \ominus e_z)\obar \#^{2k\ominus}}\in \tset$.

 \end{description}

\paragraph{Corner tiles} 
With referecnce to the  picture of Figure~\ref{fig:four-k-tiles}, let $a$ be the first symbol of $n_x$, $b$ the last symbol of $n_y$, $c$ the first symbol of $s_z$, and $d$ the last symbol of $s_t$.
\begin{center}{
\begin{description}[itemsep=2pt]
 \item $\mytile{\#}{\#}{\#}{x}\in \tsetM_2$ iff $\mytile{\#}{\#}{\#}{a}\in \tset$; \qquad $\mytile{\#}{\#}{y}{\#}\in \tsetM_2$ iff $\mytile{\#}{\#}{b}{\#}\in \tset$;
 \item $\mytile{\#}{z}{\#}{\#}\in \tsetM_2$ iff $\mytile{\#}{c}{\#}{\#}\in \tset$; \qquad $\mytile{t}{\#}{\#}{\#}\in \tsetM_2$ iff $\mytile{d}{\#}{\#}{\#}\in \tset$.
 \end{description}
}
\end{center}

\par\noindent
 Let $L(\tsetM_2)\subseteq B_k^{++}$ be the local language defined by the tile set $\tsetM_2$. From the above construction, one derives that
$
 R=\pi_k(L(\tsetM_2)).
$
\qed
\par
\medskip
Figure~\ref{fig-M} shows some examples of tiles in $\tsetM_2$, composed of four symbols in $B_3$ obtained from picture $p_0$ of Figure~\ref{fig-local-padding}. 
Each symbol in $B_3$ is represented as a picture of size $(3,3)$ where the only pixel not belonging to the frame is replaced by a white space. An example of picture in $L(\tsetM_2)$ is shown in Figure~\ref{fig-L-M}. 

\newcommand{\PreserveBackslash}[1]{\let\temp=\\#1\let\\=\temp}
\newcolumntype{C}[1]{>{\PreserveBackslash\centering}p{#1}}
\newcolumntype{R}[1]{>{\PreserveBackslash\raggedleft}p{#1}}
\newcolumntype{L}[1]{>{\PreserveBackslash\raggedright}p{#1}}

\begin{figure}
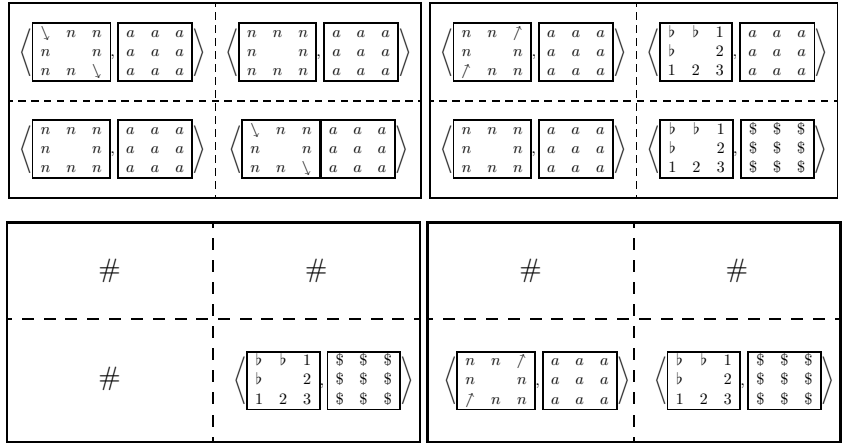
 

	\begin{center}\scalebox{0.6}{	\begin{tabular}{|c:c|} \hline &\\
			$\left \langle\begin{array}{|ccc|} \hline
				\ssearrow & n & n \\
				n & &n \\
				n & n & \ssearrow\\ \hline
			\end{array}, 
			\begin{array}{|ccc|} \hline
				a & a & a \\
				a & a&a \\
				a & a & a \\ \hline
			\end{array} 
			\right\rangle$ &
			$\left \langle\begin{array}{|ccc|} \hline
				n & n & n \\
				n & &n \\
				n & n & n\\ \hline
			\end{array},
			\begin{array}{|ccc|} \hline
				a & a & a \\
				a & a&a \\
				a & a & a \\ \hline
			\end{array} 
			\right\rangle$ \\ & \\ \hdashline & \\
			$\left \langle\begin{array}{|ccc|} \hline
				n & n & n \\
				n & &n \\
				n & n & n\\ \hline
			\end{array},
			\begin{array}{|ccc|} \hline
				a & a & a \\
				a & a&a \\
				a & a & a \\ \hline
			\end{array} 
			\right\rangle$ &
			$\left \langle\begin{array}{|ccc|} \hline
				\ssearrow & n & n \\
				n & &n \\
				n & n & \ssearrow\\ \hline
			\end{array} 
			\begin{array}{|ccc|} \hline
				a & a & a \\
				a & a&a \\
				a & a & a \\ \hline
			\end{array} 
			\right\rangle$ \\ & \\ \hline
					\end{tabular}	
			}\scalebox{0.6}{
		\begin{tabular}{|c:c|} \hline &\\
			$\left \langle\begin{array}{|ccc|} \hline
				n & n & \nnearrow \\
				n & &n \\
				\nnearrow & n & n\\ \hline
			\end{array}, 
			\begin{array}{|ccc|} \hline
				a & a & a \\
				a & a&a \\
				a & a & a \\ \hline
			\end{array} 
			\right\rangle$ &
			$\left \langle\begin{array}{| ccc|} \hline
				\flat&\flat & 1\\
				\flat & &2 \\
				1 & 2 & 3\\ \hline
			\end{array}, 
			\begin{array}{|ccc|} \hline
				a & a & a \\
				a & a&a \\
				a & a & a \\ \hline
			\end{array} 
			\right\rangle$ \\ & \\ \hdashline & \\
			$\left \langle\begin{array}{|ccc|} \hline
				n & n & n \\
				n & &n \\
				n & n & n\\ \hline
			\end{array},
			\begin{array}{|ccc|} \hline
				a & a & a \\
				a & a&a \\
				a & a & a \\ \hline
			\end{array} 
			\right\rangle$ &
			$	\left \langle\begin{array}{| ccc|} \hline
				\flat&\flat & 1\\
				\flat & &2 \\
				1 & 2 & 3\\ \hline
			\end{array}, 
			\begin{array}{|ccc|} \hline
				\$ & \$ & 
				\$ \\
				\$ & \$&\$\\
				\$ & \$ & \$ \\ \hline
			\end{array} 
			\right\rangle$ \\ & \\ \hline
		\end{tabular}	
	}
	\end{center}
\begin{center}
\begin{tabular}{|C{2.28cm}:C{2.28cm}|} \hline &\\
			 $\#$ & $\#$\\ & \\ \hdashline & \\
			 $\#$ &
		\scalebox{0.6}{	$\left \langle\begin{array}{| ccc|} \hline
				\flat&\flat & 1\\
				\flat & &2 \\
				1 & 2 & 3\\ \hline
			\end{array}, 
			\begin{array}{|ccc|} \hline
				\$ & \$ & 
				\$ \\
				\$ & \$&\$\\
				\$ & \$ & \$ \\ \hline
			\end{array} 
			\right\rangle$ }\\ & \\ \hline 
		\end{tabular}~\begin{tabular}{|C{2.28cm}:C{2.28cm}|} \hline &\\
			 $\#$ & $\#$\\ & \\ \hdashline & \\
			\scalebox{0.6}{$\left \langle\begin{array}{|ccc|} \hline
				n & n & \nnearrow \\
				n & &n \\
				\nnearrow & n & n\\ \hline
			\end{array}, 
			\begin{array}{|ccc|} \hline
				a & a & a \\
				a & a&a \\
				a & a & a \\ \hline
			\end{array} 
			\right\rangle$} &
		\scalebox{0.6}{	$\left \langle\begin{array}{| ccc|} \hline
				\flat&\flat & 1\\
				\flat & &2 \\
				1 & 2 & 3\\ \hline
			\end{array}, 
			\begin{array}{|ccc|} \hline
				\$ & \$ & 
				\$ \\
				\$ & \$&\$\\
				\$ & \$ & \$ \\ \hline
			\end{array} 
			\right\rangle$ }\\ & \\ \hline 
		\end{tabular}
\end{center}

	\caption{Two internal tiles (top), and one corner tile and one border tile(bottom) in $\tsetM_2$, each one composed of symbols in $B_k$, obtained from the tesselation of picture $p_0$ of Figure~\ref{fig-local-padding} according to Lemma~\ref{lemma-main}.}\label{fig-M}
\end{figure}

 \begin{figure}
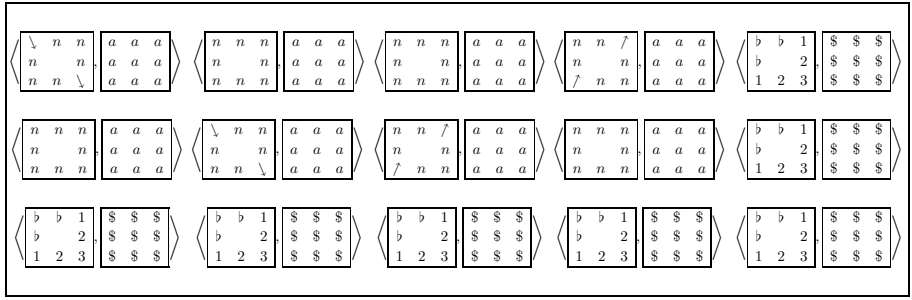

\setlength{\tabcolsep}{1pt} 
\renewcommand{\arraystretch}{1.1}
\begin{center} %Picture in ${L(\tsetM_2)}$ \\ 
\scalebox{0.80}{
	 	\begin{tabular}{|ccccc|} \hline &&&&\\
		\scalebox{0.7}{$\left \langle\begin{array}{|ccc|} \hline
				\ssearrow & n & n \\
				n & &n \\
				n & n & \ssearrow\\ \hline
			\end{array}, 
		\begin{array}{|ccc|} \hline
		a & a & a \\
		a & a&a \\
		a & a & a \\ \hline 
		\end{array} 
		\right\rangle$ }&
		\scalebox{0.7}{$\left \langle\begin{array}{|ccc|} \hline
			n & n & n \\
			n & &n \\
			n & n & n\\ \hline
		\end{array},
		\begin{array}{|ccc|} \hline
			a & a & a \\
			a & a&a \\
			a & a & a \\ \hline
		\end{array} 
		\right\rangle$} 
		& \scalebox{0.7}{$\left \langle\begin{array}{|ccc|} \hline
				n & n & n \\
				n & &n \\
				n & n & n\\ \hline
			\end{array}, 
			\begin{array}{|ccc|} \hline
				a & a & a \\
				a & a&a \\
				a & a & a \\ \hline
			\end{array} 
			\right\rangle$}
 & \scalebox{0.7}{$\left \langle\begin{array}{|ccc|} \hline
				n & n & \nnearrow \\
				n & &n \\
				\nnearrow & n & n\\ \hline
			\end{array}, 
			\begin{array}{|ccc|} \hline
				a & a & a \\
				a & a&a \\
				a & a & a \\ \hline
			\end{array} 
			\right\rangle$}
 & \scalebox{0.7}{
				$\left \langle\begin{array}{| ccc|} \hline
				\flat&\flat & 1\\
				\flat & &2 \\
				1 & 2 & 3\\ \hline
			\end{array}, 
			\begin{array}{|ccc|} \hline
				\$ & \$ & 
				\$ \\
				\$ & \$&\$\\
				\$ & \$ & \$ \\ \hline
			\end{array} 
			\right\rangle$ }
		
		\\ &&&&\\% & \\ %\hline & \\
		\scalebox{0.7}{$\left \langle\begin{array}{|ccc|} \hline
			n & n & n \\
			n & &n \\
			n & n & n\\ \hline
		\end{array},
		\begin{array}{|ccc|} \hline
			a & a & a \\
			a & a&a \\
			a & a & a \\ \hline
		\end{array} 
		\right\rangle$} &
		\scalebox{0.7}{$\left \langle\begin{array}{|ccc|} \hline
			\ssearrow & n & n \\
			n & &n \\
			n & n & \ssearrow\\ \hline
		\end{array}, 
		\begin{array}{|ccc|} \hline
			a & a & a \\
			a & a&a \\
			a & a & a \\ \hline
		\end{array} 
		\right\rangle$ 
		} 
		& \scalebox{0.7}{$\left \langle\begin{array}{|ccc|} \hline
				n & n & \nnearrow \\
				n & &n \\
				\nnearrow & n & n\\ \hline
			\end{array}, 
			\begin{array}{|ccc|} \hline
				a & a & a \\
				a & a&a \\
				a & a & a \\ \hline
			\end{array} 
			\right\rangle$}
			& \scalebox{0.7}{$\left \langle\begin{array}{|ccc|} \hline
				n & n & n \\
				n & &n \\
				n & n & n\\ \hline
			\end{array}, 
			\begin{array}{|ccc|} \hline
				a & a & a \\
				a & a&a \\
				a & a & a \\ \hline
			\end{array} 
			\right\rangle$}
 & \scalebox{0.7}{
				$\left \langle\begin{array}{| ccc|} \hline
				\flat&\flat & 1\\
				\flat & &2 \\
				1 & 2 & 3\\ \hline
			\end{array}, 
			\begin{array}{|ccc|} \hline
				\$ & \$ & 
				\$ \\
				\$ & \$&\$\\
				\$ & \$ & \$ \\ \hline
			\end{array} 
			\right\rangle$ }
		\\ &&&&\\
		\scalebox{0.7}{
				$\left \langle\begin{array}{| ccc|} \hline
				\flat&\flat & 1\\
				\flat & &2 \\
				1 & 2 & 3\\ \hline
			\end{array}, 
			\begin{array}{|ccc|} \hline
				\$ & \$ & 
				\$ \\
				\$ & \$&\$\\
				\$ & \$ & \$ \\ \hline
			\end{array} 
			\right\rangle$ }
				 %\hline
				& \scalebox{0.7}{
				$\left \langle\begin{array}{| ccc|} \hline
				\flat&\flat & 1\\
				\flat & &2 \\
				1 & 2 & 3\\ \hline
			\end{array}, 
			\begin{array}{|ccc|} \hline
				\$ & \$ & 
				\$ \\
				\$ & \$&\$\\
				\$ & \$ & \$ \\ \hline
			\end{array} 
			\right\rangle$ 
			
			} & \scalebox{0.7}{
				$\left \langle\begin{array}{| ccc|} \hline
				\flat&\flat & 1\\
				\flat & &2 \\
				1 & 2 & 3\\ \hline
			\end{array}, 
			\begin{array}{|ccc|} \hline
				\$ & \$ & 
				\$ \\
				\$ & \$&\$\\
				\$ & \$ & \$ \\ \hline
			\end{array} 
			\right\rangle$ }
			& \scalebox{0.7}{
				$\left \langle\begin{array}{| ccc|} \hline
				\flat&\flat & 1\\
				\flat & &2 \\
				1 & 2 & 3\\ \hline
			\end{array}, 
			\begin{array}{|ccc|} \hline
				\$ & \$ & 
				\$ \\
				\$ & \$&\$\\
				\$ & \$ & \$ \\ \hline
			\end{array} 
			\right\rangle$ }
			& \scalebox{0.7}{
				$\left \langle\begin{array}{| ccc|} \hline
				\flat&\flat & 1\\
				\flat & &2 \\
				1 & 2 & 3\\ \hline
			\end{array}, 
			\begin{array}{|ccc|} \hline
				\$ & \$ & 
				\$ \\
				\$ & \$&\$\\
				\$ & \$ & \$ \\ \hline
			\end{array} 
			\right\rangle$ }
			\\&&&& \\ \hline 
		\end{tabular}

		}
		
		 \caption{A picture $q_0$ in the local language $L(\tsetM_2)$ as defined in Lemma~\ref{lemma-main}, corresponding to picture $p_0$ of Figure~\ref{fig-local-padding}. The picture has size $(3,5)$, while $p_0$ has size $(9,15)$.}
\label{fig-L-M}
\end{center}

\end{figure}

\paragraph{Comma-free picture codes with two components}
For every $k\ge 2$, we define a comma-free picture code $Z\subseteq (\{0,1\}\times \Sigma_\$)^{k,k}$ to be the composition of a comma-free binary picture code $X\subseteq \{0,1\}^{k,k}$ with a set $W$ of pictures
	in $\Sigma_\$^{k,k}$: 
	\begin{equation}\label{eq:composition}
		Z= X \otimes W.
	\end{equation}
	where the operator $\otimes$ merges two isometric pictures into one, symbol by symbol.
	\medskip\par\noindent
	For instance, if $u = \mytile{1}{0}{0}{0}$ and $y= \mytile{a}{b}{b}{a}$, then $u \otimes y = 
	\mytileDynamic{\langle 1,a\rangle}{\langle 0,b\rangle}{\langle 0,b\rangle}{\langle 0,a\rangle}$.
	
	\par\noindent
	The operator can be immediately extended to a pair of sets of isometric pictures. 
	\par\noindent

\begin{figure}
\begin{center}
\scalebox{0.7}{$
	\begin{array}{cccc}
	{B_3} & \pi_3(B_3) & X_3 & Z_3 \\
\left \langle\begin{array}{|ccc|} \hline
	\ssearrow & n & n \\
	n & &n \\
	n & n & \ssearrow\\ \hline
\end{array}, 
\begin{array}{|ccc|} \hline
	a & a & a \\
	a & a&a \\
	a & a & a \\ \hline
\end{array} 
\right\rangle
& 	 	\begin{array}{| ccc|} \hline
			a & a & a\\
			a & a &a \\
			a & a & a \\ \hline
		\end{array} &
	
		\begin{array}{|ccc|} \hline
			0 & 0 & 0 
			\\
			1 & 1&0 
			\\
			1 & 1 &0 
			\\ \hline
		\end{array}&
		\begin{array}{|ccc|} \hline
			(0,a) & (0,a) & (0,a) 
			\\
			(1,a) & (1,a) &(0,a) 
			\\
			(1,a) & (1,a) &(0,a) 
			\\ \hline
		\end{array} 
		\\ \\
		\left \langle\begin{array}{|ccc|} \hline
			n & n & n \\
			n & &n \\
			n & n & n\\ \hline
		\end{array}, 
		\begin{array}{|ccc|} \hline
			a & a & a \\
			a & a&a \\
			a & a & a \\ \hline
		\end{array} 
		\right\rangle& 
		\begin{array}{| ccc|} \hline
		a & a & a\\
		a & a &a \\
		a & a & a \\ \hline
	\end{array}& 
		\begin{array}{|ccc|} \hline
			0 & 0 & 1 
			\\
			1 & 1&0 
			\\
			1 & 1 &0 
			\\ \hline
		\end{array}&
		\begin{array}{|ccc|} \hline
			(0,a) & (0,a) & (1,a) 
			\\
			(1,a) & (1,a) &(0,a) 
			\\
			(1,a) & (1,a) &(0,a) 
			\\ \hline
		\end{array} 
		\\ \\
			\left \langle\begin{array}{|ccc|} \hline
			n & n & \nnearrow \\
			n & &n \\
			\nnearrow & n & n\\ \hline
		\end{array}, 
		\begin{array}{|ccc|} \hline
			a & a & a \\
			a & a&a \\
			a & a & a \\ \hline
		\end{array} 
		\right\rangle& 
		\begin{array}{| ccc|} \hline
			a & a & a\\
			a & a &a \\
			a & a & a \\ \hline
		\end{array} &	
		\begin{array}{|ccc|} \hline
			0 & 1 & 0 
			\\
			1 & 1&0 
			\\
			1 & 1 &0 
			\\ \hline
		\end{array}&
		\begin{array}{|ccc|} \hline
			(0,a) & (1,a) & (0,a) 
			\\
			(1,a) & (1,a) &(0,a) 
			\\
			(1,a) & (1,a) &(0,a) 
			\\ \hline
		\end{array} 
		\\ \\ 
			\left \langle\begin{array}{| ccc|} \hline
				\flat&\flat & 1\\
				\flat & &2 \\
				1 & 2 & 3\\ \hline
			\end{array}, 
		\begin{array}{| ccc|} \hline
			\$ & \$ & \$\\
			\$ & \$ &\$ \\
			\$ & \$ & \$ \\ \hline
		\end{array} 
		\right\rangle& 
		\begin{array}{| ccc|} \hline
		\$ & \$ & \$\\
		\$ & \$ &\$ \\
		\$ & \$ & \$ \\ \hline
	\end{array}& 
		\begin{array}{|ccc|} \hline
			0 & 1 & 1 
			\\
			1 & 1&0 
			\\
			1 & 1 &0 
			\\ \hline
		\end{array}&
		\begin{array}{|ccc|} \hline 
			(0,\$) & (1,\$) & (1,\$)
			\\
			(1,\$) & (1,\$)&(0,\$) 
			\\
			(1,\$)& (1,\$) &(0,\$) 
			\\ \hline
		\end{array} 
	\end{array}
	$
}\end{center}
\caption{Definition of $Z_3\subseteq X_3 \otimes \Sigma_\$^{3,3}$ 
	by applying the comma-free picture code $X_3$ 	of Example~\ref{ex-commafree-family} to the image under the morphism $\pi_3$ 
	of the four elements of $B_3$ obtained from the 3-tiles of Figure~\ref{fig-3tiles}.
}\label{fig-Z3}
\end{figure}

A key point in the proof of next lemma is that, 
as a consequence of Proposition~\ref{prop-commafree-family},
 for a sufficiently large integer $k$, there exists a comma-free code 
 $X \subset \{0,1\}^{(k,k)}$ having cardinality greater than that of 
the set of ``frames'' of elements of $\kblock{L(\tset )}{k}$. This allows to 
encode the elements of the alphabet $B_k$ , defined in Lemma 1, 
by a comma-free code { $Z \subset \left(\{0,1\}\times \Sigma \right)^{(k,k)}$}, and then,
 according to Theorem 2, to express the language $R^{(k)}$ as the
 projection of a $2k$-SLT language.

 \begin{lemma}\label{lemma-main2}
		For any $R\subseteq \Sigma^{++}$ in \rec, 
		there exist 
		$k\ge 2$, a $2k$-SLT language $L$ over an alphabet $\Lambda$ 
		and a projection $\rho:\Lambda\to\Sigma$, with $|\Lambda|=2+2\cdot |\Sigma|$, such that the padded language $R^{(k)}$ can be expressed as $R^{(k)}=\rho(L)$. 
		Moreover, if $n$ is the size of the local alphabet of a tiling system recognizing $R$, then the value $k$ is $\mathcal{O}(\lg n $).
		\end{lemma} 
		\proof
\par\noindent
Define as in the proof of Lemma~\ref{lemma-main}  the finite alphabet $B_k\subseteq \Gamma^{4k}\times \Sigma_\$^{k,k}$, the local set $\tsetM_2$ of 2-tiles over $B_k \cup \{\#\}$ and the morphism $\pi_k:B_k^{++}\to\Sigma_\$^{++}$.
Hence, $R^{(k)}=\pi_k(L(\tsetM_2))$.
Consider again the set of frames $Q_k$  defined in Eq.~\eqref{eq-Qk} in the proof of Lemma~\ref{lemma-main}.

\par\noindent
 We  need to show that, for a sufficiently large integer $k$,
 there exist a comma-free picture code 
 $X\subseteq \{0,1\}^{k,k}$ such that we can define a code $Z$ in $X \otimes \pi_k(B_k)$, with $|Z|\ge |\pi_k(B_k)|$, i.e., we can associate a distinct picture code in $Z$ with each $k$-tile in $\pi_k(B_k)$. 
 \par\noindent
 Since $|Z|>|X|$, it is enough to show that $|X| \ge |Q_k|$, i.e., different code-pictures in $X$ can be assigned to $k$-tiles having different frames. 
 In fact, by definition of $B_k$, if two different $k$-tiles $\beta,\beta'\in B_k$ have the same frame, then it must be $\pi_k(\beta) \neq \pi_k(\beta')$, 
 hence by assigning the same picture code $x$ to both $\beta,\beta'$, we obtain $x\otimes \pi_k(\beta) \, \neq \, x\otimes \pi_k(\beta')$. 
 \par\noindent
Let the tiling system $(\Sigma_\$, \Gamma, \tset,  \pi)$ recognizing $R^{(k)}$ be defined as in the proof of Theorem~\ref{thm-size-kSLT-padded}. 
Thus, $\Gamma=\Gamma'\cup \Delta$, where $\Delta=\{\flat, 1, \dots, k\}$ is an alphabet disjoint from $\Gamma'$.
\par\noindent
We first show that there exists $X$ with enough code-pictures to encode all elements of the restriction $\widetilde Q_k=Q_k\cap (\Gamma -\Delta)^{4k}$, 
i.e., the set of frames that have no padding symbol; therefore, $Z$ can encode all elements of $B_k$ whose frames are in $\widetilde Q_k$.

\par\noindent
The cardinality of $\widetilde Q_k$ is at most $|\Gamma|^{4k}$. By posing $m= |\Gamma|^4$, we have $|\widetilde Q_k|\le m^k$. 
By Proposition~\ref{prop-commafreeEnough}, it follows that there exist $k \in \mathcal{O}(\lg m)=\mathcal{O}(\lg |\Gamma|)$ and 
a comma-free picture code $X \subseteq \{0,1\}^{k,k}$, such that $|X|\ge m^k\ge |\widetilde Q_k|$. 
\par\noindent
If we define $Z = X \otimes \Sigma^{k,k}$, then $|Z|\ge |X| \ge |\widetilde Q_k|$.
\medskip
\par\noindent
It remains to show that the relation $|Z|\ge |Q_k|$ also holds when extending $Z$ to $\Sigma_\$$, namely for $Z\subseteq X \otimes \Sigma_\$^{k,k}$. 
We claim that $X$ can be used to encode the remaining frames in $Q_k - \widetilde Q_k$. 
Given a frame $\beta$ in $Q_k - \widetilde Q_k$, it is enough to consider the number $i$ of rows such that $e_\beta \in \Gamma^i \Delta^{k-i}$, the number $j$ such that $w_\beta \in \Delta^{k-j}\Gamma^j$, 
and the portion of the frame of $\beta$ over $\Gamma$. 
Therefore, we need $2 \lg k$ bits for the two numbers $i,j$ and $\lg |\Gamma^i\Gamma^k\Gamma^j|\le \lg|\Gamma|^{3k}$ bits for the portion of the frame over $\Gamma$. 
\par\noindent
Hence, the total number of bits needed is $2\lg k+\lg|\Gamma^{3k}|<\lg |\Gamma^{4k}|$ ($\Gamma$ is not unary). 
The claim follows since a $k$-tile $\beta \in Q_k - \widetilde Q_k$ and a $k$-tile $\beta'\in \widetilde Q_k$ are always such that $\pi_k(\beta)\neq\pi_k(\beta')$.

\medskip
\par\noindent
To encode the elements of $Q_k$, 
we consider the morphism $\codeX{\; }: Q_k \to X$
that associates each element $\gamma\in Q_k$ with a different code-picture $\codeX{\gamma}$.
\par\noindent
The comma-free picture code $Z$ in Eq.~\eqref{eq:composition} is then defined by morphism $\code{Z}{\;}: B_k^{++}\to (\{0,1\}\times \Sigma_\$)^{++}$ associating 
each $x = (\gamma,p)$ with the code-picture $\code{Z}{x}= \codeX{\gamma}\otimes p$. An example is in Figure~\ref{fig-Z3}, which shows the definition of a comma-free picture code $Z_3$ by application of a comma-free picture code $X_3\subseteq \{0,1\}^{3,3}$ to the image under $\pi_3$ of elements of the alphabet $B_3$, obtained from the $3$-tiles of Figure~\ref{fig-3tiles}.
\par\noindent
We set the alphabet $\Lambda$ of the statement to $\{0,1\}\times \Sigma_\$.$ By Theorem~\ref{th-SLTCode}, the language
$\code{Z}{L(\tsetM_2)}\subseteq \Lambda^{++}$ is $2k$-SLT. 
{ Let $\tsetM_{2k}$ denote the set of $2k$-tiles over the alphabet $\Lambda\cup\{\#\}$ defining the language $\code{Z}{L(\tsetM_2)}$. 
An example of a picture in $L(\tsetM_{2k})$ is given in Figure~\ref{fig-6-tiles}.
\medskip
\par\noindent
Let $\rho: \Lambda\to\Sigma_\$$ be the projection of each element $(b,a)\in \{0,1\}\times \Sigma_\$$ to its second component $a$. 
We prove that %\begin{equation}
 $\rho\left(\code{Z}{L(\tsetM_2)}\right)=R^{(k)}.$
\par\noindent
 If $p\in R^{(k)}$, with $p$ of size $(kr,kc)$ with $r,c\ge 1$, then there exists a picture $q \in L(\tset )$ such that $\pi(q)=p$. Let $\widetilde q = \code{Z}{q}$. By definition of $Z$, 
 $\rho(\widetilde q) = \pi(q)=p$, hence $p=\rho(\code{Z}{q}) \in \rho\left(\code{Z}{L(\tsetM_2)}\right)$.
 \par\noindent If $p\in \rho\left(\code{Z}{L(\tsetM_2)\right)}$, with $p$ of size $(kr,kc)$ for some $r,c\ge 1$, then 
 there is $q\in L(\tsetM_2)$ such that $p = \rho(\code{Z}{q})$. By definition of $\rho$ and $\pi$, we have $\pi(q)=p$, i.e., 
 $p \in R^{(k)}$.
 \qed

\begin{figure}
\setlength{\tabcolsep}{1pt} 
\renewcommand{\arraystretch}{1.1}
\begin{center} %Picture in ${L(\tsetM_2)}$ \\ 
\scalebox{0.7}{
	 	\begin{tabular}{|c:c:c:c:c|}\hline 
		$\begin{array}{ccc} 
			(0,a) & (0,a) & (0,a) 
			\\
			(1,a) & (1,a) &(0,a) 
			\\
			(1,a) & (1,a) &(0,a) \\
		\end{array}$ &
		$\begin{array}{ccc} 
			(0,a) & (0,a) & (1,a) 
			\\
			(1,a) & (1,a) &(0,a) 
			\\
			(1,a) & (1,a) &(0,a) 
			\\ 
		\end{array}$ 
		& $\begin{array}{ccc} 
			(0,a) & (0,a) & (1,a) 
			\\
			(1,a) & (1,a) &(0,a) 
			\\
			(1,a) & (1,a) &(0,a) 
			\\ 
		\end{array}$
 & $\begin{array}{ccc} 
			(0,a) & (1,a) & (0,a) 
			\\
			(1,a) & (1,a) &(0,a) 
			\\
			(1,a) & (1,a) &(0,a) \\
		\end{array}$
 & $\begin{array}{ccc} 
			(0,\$) & (1,\$) & (1,\$)
			\\
			(1,\$) & (1,\$)&(0,\$) 
			\\
			(1,\$)& (1,\$) &(0,\$) 
			\\ 
		\end{array} $
		
		\\ \hdashline % & \\ %\hline & \\
		$\begin{array}{ccc} 
			(0,a) & (0,a) & (1,a) 
			\\
			(1,a) & (1,a) &(0,a) 
			\\
			(1,a) & (1,a) &(0,a) 
			\\ 
		\end{array}$ &
		$\begin{array}{ccc} 
			(0,a) & (0,a) & (0,a) 
			\\
			(1,a) & (1,a) &(0,a) 
			\\
			(1,a) & (1,a) &(0,a) \\
		\end{array}$
		& $\begin{array}{ccc} 
			(0,a) & (1,a) & (0,a) 
			\\
			(1,a) & (1,a) &(0,a) 
			\\
			(1,a) & (1,a) &(0,a) 
			\\ 
		\end{array}$
			& $\begin{array}{ccc} 
			(0,a) & (0,a) & (1,a) 
			\\
			(1,a) & (1,a) &(0,a) 
			\\
			(1,a) & (1,a) &(0,a) 
			\\ 
		\end{array}$
 & $\begin{array}{ccc} 
			(0,\$) & (1,\$) & (1,\$)
			\\
			(1,\$) & (1,\$)&(0,\$) 
			\\
			(1,\$)& (1,\$) &(0,\$) 
			\\ 
		\end{array} $
		\\ \hdashline $\begin{array}{ccc} 
			(0,\$) & (1,\$) & (1,\$)
			\\
			(1,\$) & (1,\$)&(0,\$) 
			\\
			(1,\$)& (1,\$) &(0,\$) 
			\\ 
		\end{array} $
				 %\hline
				& $\begin{array}{ccc} 
			(0,\$) & (1,\$) & (1,\$)
			\\
			(1,\$) & (1,\$)&(0,\$) 
			\\
			(1,\$)& (1,\$) &(0,\$) 
		\end{array} $
		& $\begin{array}{ccc} 
			(0,\$) & (1,\$) & (1,\$)
			\\
			(1,\$) & (1,\$)&(0,\$) 
			\\
			(1,\$)& (1,\$) &(0,\$) 
			\\ 
		\end{array} $
			& $\begin{array}{ccc} 
			(0,\$) & (1,\$) & (1,\$)
			\\
			(1,\$) & (1,\$)&(0,\$) 
			\\
			(1,\$)& (1,\$) &(0,\$) 
			\\ 
		\end{array} $
			& $\begin{array}{ccc} 
			(0,\$) & (1,\$) & (1,\$)
			\\
			(1,\$) & (1,\$)&(0,\$) 
			\\
			(1,\$)& (1,\$) &(0,\$) 
			\\ 
		\end{array} $
			\\ \hline
		\end{tabular}	
		}
		
		 \caption{The picture $\code{Z}{q_0}$ in the language $L(M_6)$, where $q_0$ is the picture of Figure~\ref{fig-L-M}. The picture has size $(9,15)$ and its projection to $\Sigma$ is picture $p_0$ of Figure~\ref{fig-local-padding}, i.e., $\rho(\code{Z}{q_0})=p_0$.}
\label{fig-6-tiles}
\end{center}

\end{figure}

\medskip
\par
Next, we prove the final result.
\begin{theorem}\label{thm-main}
	For any $R\subseteq \Sigma^{++}$ in \rec, 
	there exist 
	$k\ge 2$ and a $2k$-tiling system with alphabetic ratio 2, recognizing $R$. 
Moreover, if $n$ is the size of the local alphabet of a tiling system recognizing $R$, then the value $k$ is $\mathcal{O}(\lg n$).
\end{theorem}
\proof Let $\Lambda, L, \rho$ be defined as in Lemma~\ref{lemma-main2}, and let $\tsetM_{2k}$ be the set of $2k$-tiles defining the $2k$-SLT language $L$.
Let $\Lambda_{\$} \subseteq \Lambda$ be the set 
$\{ (0,\$), (1,\$) \}$, and let $\Theta$ be $\Lambda - \Lambda_{\$}$. Hence, $|\Theta| = 2 |\Sigma_\$|-2= 2|\Sigma|$.
Define the projection $\pi:\Theta\to \Sigma$ as the restriction of $\rho$ induced by the subset $\Theta$ of $\Lambda$.
\par\noindent
 For simplicity, we assume in the following that every picture of $R$ has both horizontal and vertical sizes greater or equal to $2k-1$. The cases when one or both dimensions are smaller are analogous, requiring only to consider a larger border. 
 \par\noindent
 We define a set of $2k$-tiles $\tsetM'$ over alphabet $\Theta$, defining a $2k$-SLT language $L(\tsetM')$, in three steps: 
	\begin{enumerate}
	\item Delete, from the set $\tsetM_{2k}$, the $2k$-tiles $\theta$ having at least one of the two  forms:
	\begin{itemize}
	 \item 
	$\theta = s\obar z \obar z'$ with $s \in \Lambda^{2k,2k-2}, z\in \Lambda_\$^{\ominus 2k} $, $z' \in \Lambda_\$^{\ominus 2k} \cup \#^{\ominus 2k}$ (i.e, having at least two columns containing elements in $\Lambda_\$ $, 
	or one column in $\Lambda_\$$ and one with $\#$ ); 
	\item $\theta = s\ominus z \ominus z'$ with $s \in \Lambda^{2k,2k-2}, z\in \Lambda_\$^{\obar 2k} $, $z' \in \Lambda_\$^{\bar 2k} \cup \#^{\bar 2k}$ (i.e, having at least two rows containing elements in $\Lambda_\$ $, 
	or one row in $\Lambda_\$$ and one row with $\#$ ).
	\end{itemize}

	\item Substitute (in the set obtained after the first step) all the occurrences of the elements in $\Lambda_\$ $ with the symbol \#; 
	\end{enumerate}
{For instance, in step 1 the $2k$-tile 
	\scalebox{0.6}{$
		\begin{array}{|p{1cm}:c:c|}\hline
		& z&z' \\
		\centering\scalebox{2}{s}&&\\
		& &\\ \hline
		\end{array}
		$}
	 with $s \in (\Theta)^{2k,2k-2}, z,z' \in \Lambda_\$^{\ominus 2k} $, is deleted;
	in step 2 the $2k$-tiles: 
\scalebox{0.6}{$
\begin{array}{|p{1cm}:c|}\hline
 & z \\
 \centering\scalebox{2}{t}&\\
 & \\ \hline
 \end{array}
 $} and 
\scalebox{0.6}{$\begin{array}{|p{1cm}:c|}\hline
	& z\\
	\centering\scalebox{2}{u}&\\ \hdashline
	$z''$ & \\ \hline
	\end{array}
	$},
 with $t \in(\Theta)^{2k,2k-1}, z \in \Lambda_\$^{\ominus 2k} $, $u \in (\Theta)^{2k-1,2k-1}$ and $z''\in \Lambda_\$^{\obar 2k}$, are respectively replaced by: 
\scalebox{0.6}{$\begin{array}{|p{1cm}:c|}\hline
 & \# \\
 \centering\scalebox{2}{t}&\dots\\
 & \# \\ \hline
 \end{array}
$} and     
\scalebox{0.6}{$\begin{array}{|p{1cm}:c|}\hline
	& \# \\
	\centering\scalebox{2}{u}&\dots\\ \hdashline
	\# \dots & \# \\ \hline
	\end{array}
	$}
}
\medskip
\par\noindent
From the above construction, one can derive 
$R=\pi(L(\tsetM'))$ as follows. 

\par\noindent 
We first show that $R \subseteq \pi(L(\tsetM'))$. 
Let $p_\Sigma$ be a picture of $R$. By Definition~\ref{def-padded} of padding languages, 
there exists a picture $p\in R^{(k)}$ of the form $(p_{\Sigma}\ominus v_\$)\obar h_\$$ 
for some $v_\$ \in V_k, h_\$ \in H_k$. 
Let $q \in L(\tsetM_{2k})$ be a pre-image pf $p$, 
hence $q = (q_{\Lambda'}\ominus v)\obar h$, with $\pi(q_{\Lambda'})=p_\Sigma, \pi(v)= v_\$, \pi(h)=h_\$ $. 
The $2k$-tiles of $\tsetM_{2k}$ of the form $s \obar \Lambda_\$^{\obar 2k}$, with $s \in \Sigma^{2k-1,2k}$, are replaced in $\tsetM'$ by $2k$-tiles 
of the form $s \obar \#^{\obar 2k}$. 
\par\noindent
A similar argument can be applied to the $2k$-tiles of the form $s \ominus \Lambda_\$^{\ominus 2k}$ for $s\in \Sigma^{2k,2k-1}$ and 
to the $2k$-tiles of the form $(s \obar \Lambda_\$^{\obar 2k-1}) \ominus \Lambda_\$^{\bar 2k}$ for $s\in \Sigma^{2k-1,2k-1}$.
Since $q_{\Lambda'} \in L(\tsetM_{2k})$, all the $2k$-tiles in $\ktile{\#^{\obar +}\ominus(\#^{\ominus +}\obar q_{\Lambda'})}{2k}$ are in $\tsetM_{2k}$, hence also in $\tsetM'$. 
Therefore ${q_{\Lambda'}} \in L(\tsetM')$. Since $\pi(q_{\Lambda'})=p_\Sigma \in R$, it follows that 
$R \subseteq \pi(L(\tsetM'))$. 
\medskip
\par\noindent
We now show that $\pi(L(\tsetM')) \subseteq  R$. 
Let $p_{\Sigma}$ be a picture in $\pi(L(\tsetM'))$ and let $q_{\Lambda'}$ be a pre-image of $p_\Sigma$, hence $q_{\Lambda'}\in L(\tsetM')$. 
By construction of $\tsetM'$, each $2k$-tile $t$ of $\ktile{q_{\Lambda'}}{2k}$  is either in $\tsetM_{2k}$ or has one of the following forms: 
$r \obar \#^{\ominus 2k}$ for $r \in \Sigma^{2k,2k-1}$, or $r \ominus \#^{\obar 2k}$ for $r\in \Sigma^{2k-1,2k}$, or 
$(r \obar \#^{\ominus 2k-1}) \ominus \#^{\obar 2k}$ for $r\in \Sigma^{2k-1,2k-1}$. 
\par\noindent
Let $t\ktile{q_{\Lambda'}}{2k}q_{\Lambda'}$, with $t \not\in  \tsetM_{2k}$. Assume  that $t$ has the following form: 
$r\obar \#^{\ominus 2k}$ for some $r \in \Sigma^{2k,2k-1}$. The other possible forms of $t$ can be dealt with analogously. 
\par\noindent
We assume, without loss of generality, 
that in the comma-free picture code $Z\subseteq X\otimes \Sigma^{k,k}$, $X$  is defined according to  Definition~\ref{def-family-comma-free}; therefore, 
the subpicture $r$ of $t$ must contain a least a row $j$, for some $1\le j \le 2k$,  of the form $sxp$, where $x$ is a (horizontal) code-word and $p,s$ are, respectively, a prefix and  a suffix of horizontal code-words. 
\par\noindent
Notice that $p$ must be shorter than $k$, otherwise $|sxp|=2k-1, |x|=k$. Also, $p\neq \varepsilon$, otherwise $t\in \tsetM_{2k}$. 
Let $|p|=i$, with $1\le i <k$. Since $t$ is in $\tsetM'$, then there is a $2k$-tile $t'\in \tsetM_{2k}-\tsetM'$, such that the same row $j$ 
has the form $xy$ where $y$ is a codeword of the form $ps'$, for some $s' \in \Lambda_\$^{k-i}$. 
Therefore, $t'$ is in $\Lambda^{2k,k+i}\obar\Lambda_\$^{2k,k-i}$; $t'$ was deleted from $\tsetM_{2k}$ when 
defining $\tsetM'$, together with some ``intermediate'' $2k$-tiles in $\Lambda^{2k,k+i+1}\obar\Lambda_\$^{2k,k-i-1}$, $\Lambda^{2k,k+i+2}\obar\Lambda_\$^{2k,k-i-2}$, etc. 
\par\noindent
Hence, it is possible to identify, based on $\tsetM'$ and the definition of comma-free codes, all and only the $2k$-tiles of $\tsetM_{2k}$ which were deleted when defining $\tsetM'$. 
Hence, we can construct 
a picture $q = (q_{\Lambda'}\ominus v)\obar h$, with $\pi(q_{\Lambda'})=p_\Sigma$, for some $v,h \in \Lambda_\$^{+,+}$,
 such that 
$q \in L(\tsetM_{2k})$. Therefore, $\pi(q) \in R^{(K)}$. By definition of padding language, it follows that $p_\Sigma \in R$.
\qed

\section{Conclusion}\label{s-concl}
Our main result (Theorem~\ref{thm-main}) shows that every  recognizable picture language  is the projection with alphabetic ratio 2 of a strictly locally testable language. 
Moreover, if $n$ is the size of the local alphabet of a tiling system recognizing the language, then the order of testability is $\mathcal{O}(\lg n$).
A curious example is that any black-and-white recognizable picture is the  projection of a  strictly locally testable language on a four letter alphabet.

\par
  The proof relies on two novel results having a potential interest of their own:
in Section~\ref{sect-newFamilyCodes} a new family of 2D comma-free codes having a precise numerosity  bound, and  in Section~\ref{ssect-SLTofEncodedPict} the  property that that a picture morphism mapping letters to comma-free code-pictures, transforms local pictures into SLT ones.

\par
This result can be placed next to the similar ones for regular word languages (\emph{v.s.} Section~\ref{sect:introd} and~\cite{DBLP:journals/ijfcs/Crespi-ReghizziP12}) 
and for tree languages~\cite{DBLP:conf/lata/Crespi-Reghizzi21}. 
Altogether, they give evidence that, for three significant language families the same property, that we may call the Extended Medvedev's theorem with alphabetic ratio two, holds. 
Differently said, in the three cases the alphabetic ratio of two is sufficient and necessary to characterize a language as the morphic image of a strictly locally testable language.
\par 
Although the three cases encompass mathematical objects of quite different kinds, all of them satisfy the prerequisite that a (non-extended) Medvedev's theorem holds, which is based on a notion of locality, respectively, for words, for rectangular arrays, 
and for tree graphs. In the future, it would be interesting to check whether other families endowed with the basic Medvedev's theorem also have the above property.
\par
Beyond these cases, some loose resemblance may be seen between our result for \rec~and some studies on 2D cellular automata (a classic one is~\cite{ARSmith1971}) that study the tradeoff of two parameters: the size of the 
cell neighborhood and the size of the state set. Clearly, the first parameter is analogous to the order of $k$-tile, and the second one to the size of the cell alphabet.
However, there are of course fundamental differences between the two models.

\ifarXiv \bibliographystyle{plain}
\else \bibliographystyle{abbrv}
\section*{\refname}
\fi
\bibliography{automatabib}

\begin{thebibliography}{10}

\bibitem{DBLP:journals/tcs/AnselmoGM17}
M.~Anselmo, D.~Giammarresi, and M.~Madonia.
\newblock Non-expandable non-overlapping sets of pictures.
\newblock {\em Theor. Comput. Sci.}, 657:127--136, 2017.

\bibitem{ANSELMO20174}
M.~Anselmo and M.~Madonia.
\newblock Two-dimensional comma-free and cylindric codes.
\newblock {\em Theoretical Computer Science}, 658:4--17, 2017.
\newblock Formal Languages and Automata: Models, Methods and Application In
  honour of the 70th birthday of Antonio Restivo.

\bibitem{BerPerReu09}
J.~Berstel, D.~Perrin, and C.~Reutenauer.
\newblock {\em {Codes and Automata}}, volume 129 of {\em Encyclopedia of
  Mathematics and its Applications}.
\newblock CUP, 2009.

\bibitem{DBLP:books/ems/21/Crespi-ReghizziGL21}
S.~Crespi{-}Reghizzi, D.~Giammarresi, and V.~Lonati.
\newblock Two-dimensional models.
\newblock In J.~Pin, editor, {\em Handbook of Automata Theory}, pages 303--333.
  European Mathematical Society Publishing House, Z{\"{u}}rich, Switzerland,
  2021.

\bibitem{Crespi-ReghizziRestivoSanPietro21}
S.~Crespi{-}Reghizzi, A.~Restivo, and P.~{San Pietro}.
\newblock Reducing local alphabet size in recognizable picture languages.
\newblock In N.~Moreira and R.~Reis, editors, {\em Developments in Language
  Theory , {DLT} 2021}, volume 12811 of {\em LNCS}, pages 103--116. Springer,
  2021.

\bibitem{DBLP:journals/ijfcs/Crespi-ReghizziP12}
S.~Crespi{-}Reghizzi and P.~{San Pietro}.
\newblock From regular to strictly locally testable languages.
\newblock {\em Int. J. Found. Comput. Sci.}, 23(8):1711--1728, 2012.

\bibitem{DBLP:conf/cai/Crespi-Reghizzi19}
S.~Crespi{-}Reghizzi and P.~{San Pietro}.
\newblock Regular languages as local functions with small alphabets.
\newblock In M.~Ciric, M.~Droste, and J.~Pin, editors, {\em Algebraic
  Informatics - {CAI} 2019,}, volume 11545 of {\em Lecture Notes in Computer
  Science}, pages 124--137. Springer, 2019.

\bibitem{DBLP:conf/lata/Crespi-Reghizzi21}
S.~Crespi{-}Reghizzi and P.~{San Pietro}.
\newblock Homomorphic characterization of tree languages based on comma-free
  encoding.
\newblock In {\em {LATA} 2021}, LNCS 12638, pages 241--254. Springer, 2021.

\bibitem{DeLucaRestivo1980}
A.~de~Luca and A.~Restivo.
\newblock A characterization of strictly locally testable languages and its
  applications to subsemigroups of a free semigroup.
\newblock {\em Inf. Control.}, 44(3):300--319, 1980.

\bibitem{DBLP:journals/tit/Eastman65}
W.~L. Eastman.
\newblock On the construction of comma-free codes.
\newblock {\em {IEEE} Trans. Inf. Theory}, 11(2):263--267, 1965.

\bibitem{GAMARD201758}
G.~Gamard, G.~Richomme, J.~Shallit, and T.~J. Smith.
\newblock Periodicity in rectangular arrays.
\newblock {\em Information Processing Letters}, 118:58--63, 2017.

\bibitem{GiammarresiRestivo1992}
D.~Giammarresi and A.~Restivo.
\newblock Recognizable picture languages.
\newblock {\em Int. Journ. Pattern Recognition and Artificial Intelligence},
  6(2-3):241--256, 1992.

\bibitem{GiammRestivo1997}
D.~Giammarresi and A.~Restivo.
\newblock Two-dimensional languages.
\newblock In G.~Rozenberg and A.~Salomaa, editors, {\em Handbook of formal
  languages, vol. 3}, pages 215--267. Springer, 1997.

\bibitem{DBLP:journals/iandc/GiammarresiRST96}
D.~Giammarresi, A.~Restivo, S.~Seibert, and W.~Thomas.
\newblock Monadic second-order logic over rectangular pictures and
  recognizability by tiling systems.
\newblock {\em Inf. Comput.}, 125(1):32--45, 1996.

\bibitem{GolombEtAl58}
S.~W. Golomb, B.~Gordon, and L.~Welch.
\newblock Comma-free codes.
\newblock {\em Canad. Journ. Math.}, 10:202--209, 1958.

\bibitem{HardyWright}
G.~Hardy and E.~Wright.
\newblock {\em An Introduction to the Theory of Numbers, 6th edition}.
\newblock Oxford University Press, July 2008.

\bibitem{HashiguchiHonda1976}
K.~Hashiguchi and N.~Honda.
\newblock Properties of code events and homomorphisms over regular events.
\newblock {\em J. Comput. Syst. Sci.}, 12(3):352--367, 1976.

\bibitem{McPa71}
R.~McNaughton and S.~Papert.
\newblock {\em Counter-free Automata}.
\newblock {MIT} {P}ress, {C}ambridge, {USA}, 1971.

\bibitem{Medvedev1964}
Y.~T. Medvedev.
\newblock On the class of events representable in a finite automaton.
\newblock In E.~F. Moore, editor, {\em Sequential machines -- Selected papers},
  pages 215--227. Addison-Wesley, 1964.
\newblock Originally published in Russian in Avtomaty, 1956, pp. 385--401.

\bibitem{DBLP:journals/dm/PerrinR18}
D.~Perrin and C.~Reutenauer.
\newblock Hall sets, {Lazard} sets and comma-free codes.
\newblock {\em Discret. Math.}, 341(1):232--243, 2018.

\bibitem{ARSmith1971}
A.~{Ray Smith III}.
\newblock Cellular automata complexity trade-offs.
\newblock {\em Information and Control}, 18:466--482, 1971.

\bibitem{Restivo1974}
A.~Restivo.
\newblock On a question of {McNaughton} and {Papert}.
\newblock {\em Inf. Control.}, 25(1):93--101, 1974.

\bibitem{Eilenberg74}
{S. Eilenberg}.
\newblock {\em Automata, Languages, and Machines}, volume~A.
\newblock Academic Press, 1974.

\bibitem{DBLP:journals/jcss/Thomas82}
W.~Thomas.
\newblock Classifying regular events in symbolic logic.
\newblock {\em J. Comput. Syst. Sci.}, 25(3):360--376, 1982.

\end{thebibliography}

\end{document}